\documentclass[useAMS, usenatbib, referee]{biom}
\usepackage{amsmath, amssymb}
\usepackage{graphicx}
\usepackage{bm}
\usepackage{algorithm,algpseudocode}
\usepackage{mathabx}
\usepackage{float}
\usepackage{url}
\usepackage{color}
\usepackage[T1,OT1]{fontenc}

\newcommand{\bbeta}{ \mbox{\boldmath $\beta$}}

\newcommand{\bA}{\mbox{\bf A}}

\newcommand{\bX}{\mbox{\bf X}}
\newcommand{\bB}{\mbox{\bf B}}
\newcommand{\bZ}{\mbox{\bf Z}}

\newcommand{\bY}{\mbox{\bf Y}}

\newcommand{\bV}{\mbox{\bf V}}

\newcommand{\bI}{\mbox{\bf I}}

\newcommand{\bW}{\mbox{\bf W}}

\newcommand{\bU}{\mbox{\bf U}}
\newcommand{\bS}{\mbox{\bf S}}

\newcommand{\bk}{\mbox{\bf k}}

\newcommand{\beq}{ \begin{equation}}
\newcommand{\eeq}{ \end{equation}}
\newcommand{\beqn}{ \begin{eqnarray}}
\newcommand{\eeqn}{ \end{eqnarray}}

\usepackage{caption}
\DeclareTextCommand{\DJ}{OT1}{%
  \raisebox{-0.1ex}{\scalebox{0.75}[1.4]{--}}\kern-.4em D%
}

\title{Two-Stage Estimators for Spatial Confounding with Point-Referenced Data}

\author{Nate Wiecha$^{1,*}$\email{nbwiecha@ncsu.edu}, 
Jane A. Hoppin$^{2}$, and Brian J. Reich$^{1}$ \\
$^{1}$Department of Statistics, North Carolina State University, Raleigh, North Carolina, U.S.A \\
$^{2}$Department of Biological Sciences, North Carolina State University, Raleigh, North Carolina, U.S.A}

\begin{document}


\date{{\it Received July} 2024. {\it Revised XXXX} 20XX.  {\it
Accepted XXXX} 20XX.}



\pagerange{\pageref{firstpage}--\pageref{lastpage}} 
\volume{00}
\pubyear{2024}
\artmonth{XXXXXXX}


\doi{11.1111/j.1111-1111.1111.11111.x}


\label{firstpage}


\begin{abstract}
Public health data are often spatially dependent, but standard spatial regression methods can suffer from bias and invalid inference when the independent variable is associated with spatially-correlated residuals. This could occur if, for example, there is an unmeasured environmental contaminant associated with the independent and outcome variables in a spatial regression analysis. Geoadditive structural equation modeling (gSEM), in which an estimated spatial trend is removed from both the explanatory and response variables before estimating the parameters of interest, has previously been proposed as a solution, but there has been little investigation of gSEM's properties with point-referenced data. We link gSEM to results on double machine learning and semiparametric regression based on two-stage procedures. We propose using these semiparametric estimators for spatial regression using Gaussian processes with Mat\`ern covariance to estimate the spatial trends, and term this class of estimators Double Spatial Regression (DSR). We derive regularity conditions for root-$n$ asymptotic normality and consistency and closed-form variance estimation, and show that in simulations where standard spatial regression estimators are highly biased and have poor coverage, DSR can mitigate bias more effectively than competitors and obtain nominal coverage.
\end{abstract}

%



\maketitle


%

\section{Introduction}\label{s:intro}

\normalsize
Public health data are often observational and exhibit spatial dependence, such as in environmental contaminations that may spread over a geographic region \citep{cressiebook}. 
Spatial regression methods can improve efficiency, allow proper uncertainty quantification, and enhance predictive accuracy \citep{cressiebook}. However, association between explanatory variables and latent functions of space in the response model can cause bias in estimated regression coefficients, and invalid statistical inference due to poor uncertainty quantification \citep{reich2006effects}. This has often been termed ``spatial confounding," and was first observed in \cite{clayton:1993}, discussed further in \cite{reich2006effects} and \cite{hodges2010adding}, and studied in other papers such as \cite{paciorek2010importance}. The analogous issue in spline models has been discussed in \cite{rice1986convergence} and \cite{wood2017generalized}. Several recent papers \citep{gilbert, dupont2023demystifying, khan_berrett} have discussed definitions, causes, and effects of spatial confounding in attempts to unify varying accounts.

Methods proposed to deal with spatial confounding aim to reduce bias in linear regression parameter estimates, such as in \cite{marques_explicitly_correlated_grf}, \cite{guan_2020}, and \cite{schnell2019mitigating}. Methods similar to our proposal are geoadditive structural equation modeling \citep{thaden2018structural}, Spatial+ \citep{dupont2022spatial+}, and a shift estimand, which does not assume a linear treatment effect, studied in \cite{gilbert}. Geoadditive structural equation modeling (gSEM) subtracts estimated latent functions of space from the treatment and response, and regresses those residuals onto each other. Spatial+ subtracts an estimated function of space from the treatment variable and regresses the response onto those residuals {and a spatial term modeled by a thin-plate spline}. The shift estimand implemented in \cite{gilbert} subtracts an estimated function of space from the treatment (in order to estimate a conditional density function) and response, but without assuming a linear and additive effect of the explanatory variable. Identifiability in gSEM, Spatial+, the shift estimand, and our method is typically due to independent, non-spatial variation in the treatment and response. In this situation, \cite{gilbert2023consistency} also showed that the standard generalized least squares (GLS), spline, and Gaussian process estimates are consistent even in some cases where they are misspecified due to spatial confounding, but noted that this does not imply variance estimates are accurate, and these estimators may converge at a rate slower than $n^{-1/2}$.

gSEM was initially proposed for areal spatial data, and its analogue for point-referenced spatial data{, based on Appendix B of \cite{thaden2018structural},} has not been studied thoroughly. The subject of this paper is point-referenced spatial data, so we use ``gSEM" to refer to its {point-referenced implementation.}
\cite{dupont2022spatial+} states that ``it is not immediately clear why the method works"; 
\cite{dupont2023demystifying} states that ``the bias reduction will only be successful under the assumption that the initial regressions successfully remove all spatial confounders." gSEM has also lacked a closed-form variance estimate.

Literature on semiparametric regression \citep{robinson1988root, andrews1994asymptotics, dml} proves that under spatial confounding and additional regularity conditions, a class of estimators, including some nearly identical to gSEM \citep{robinson1988root}, can achieve $n^{-1/2}$-consistency and asymptotic normality, with a consistent closed-form variance estimate, even when not removing spatial confounding {fully} in initial regression estimates of spatial trends. Broadly speaking, these estimators first estimate the latent functions of space using nonparametric regression, subtract these estimates from the observed variables, and use those residuals {to estimate the regression parameter in a second stage}; the preliminary estimates of the latent functions need converge to the true values only slowly. Due to a form of orthogonality between {the estimator} and the preliminary estimates of the latent functions, the asymptotic variance of {the estimator} is unaffected by using these preliminary estimates. By either using sample-splitting as in \cite{dml} or requiring additional smoothness conditions as in \cite{andrews1994asymptotics}, asymptotic bias due to overfitting is controlled. A more detailed overview of the related semiparametric theory is presented in the Supplementary Materials section S1.

{\cite{gilbert} showed that a Double Machine Learning (DML) estimator can address spatial confounding using a purely nonparametric estimator of the average effect of a $\delta$-shift in a scalar treatment variable. We show that a semiparametric DML approach to spatial confounding, while restrictive in some ways, has substantial upsides. \cite{gilbert} relies on theory originally meant for social network data \citep{Ogburn02012024}, which with spatial data results in an unnecessarily restricted sample-splitting strategy whose validity is unproven. No closed-form variance estimate is available for the shift estimator. In our semiparametric approach, we are able to use fully random sample splits, which allows sharper, clearer regularity conditions due to better estimation of spatial trends. The parametric model for treatment effects allows closed-form variance estimation using standard empirical estimators, avoids having to specify a shift $\delta$, and easily incorporates multiple treatment variables. Our approach avoids estimating a density function, sometimes a computational difficulty. Finally, most methods for spatial confounding and spatial regression assume a linear model, so we continue to improve understanding of this common model.}

We apply these semiparametric estimators to spatial confounding scenarios, using Gaussian Process (GP) regression to estimate the latent functions of space, { a procedure also known as Kriging}. We primarily rely on methods and theory from \cite{dml}{, extending \citeauthor{dml}'s partially linear model to a vector treatment variable}. As the two-stage estimators are termed Double Machine Learning in \cite{dml}, we refer to their narrower use to address spatial confounding as Double Spatial Regression (DSR). In simulations of spatial confounding, we show that DSR can provide superior performance in severe confounding scenarios. Finally, we analyze the association between of five per- and polyfluoroaklyl substances (PFAS) and thyroid stimulating hormone (TSH) levels in blood using DSR.

\section{Double Spatial Regression}\label{s:SDML}
\subsection{Model and notation}
For $i\in\{1,...,n\}$, let $Y_i\in \mathbb{R}$ be the response variable, $\bA_i = (a_{i1} , a_{i2}, ... , a_{i\ell})^T \in \mathbb{R}^{\ell}$ be the treatment variables, $\bZ_i = (z_{i1}, z_{i2}, ..., z_{iv})^T \in \mathbb{R}^v$ be the covariates, and $\bS_i$ be the spatial location contained in spatial domain $\mathcal{S}\subset \mathbb{R}^d$. 
The assumed model for $j = 1, ..., \ell$ is:
\begin{align}\label{eq:SDML_practical_model1}
    Y_i &= \bA_i^T\bbeta_0 + \bZ_i^T\boldsymbol{\theta}_{00} + g_0(\bS_i) + U_i\\
    A_{ij} &= \bZ_i^T\boldsymbol{\theta}_{0j} + m_{0j}(\bS_i) + V_{ij}, \label{eq:SDML_practical_model2}
\end{align}
where $\bbeta_0$ is the regression coefficient vector of interest. The vectors $\boldsymbol{\theta}_{00}, ..., \boldsymbol{\theta}_{0\ell}$ and the vector of functions $\eta_0=(g_0, m_{01}, ..., m_{0\ell})$ are all nuisance parameters. $U_i$ and $V_{ij}$ are {random} error terms such that $E(U_i | \bA_i, \bZ_i, \bS_i) = E(V_{ij} |\bZ_i, \bS_i) = 0$. 
{Spatial confounding can be thought of as the presence of an unmeasured confounder with spatial structure, which is stochastically dependent with $\bA_i$, and which affects the conditional expectation of $Y_i$ \citep{thaden2018structural}. Spatial confounding is present in (\ref{eq:SDML_practical_model1}) and (\ref{eq:SDML_practical_model2}) unless the functions $g_0(\bS_i)$ and $m_{0j}(\bS_i)$ are orthogonal with respect to the distribution of $\bS_i$. If they are not orthogonal, standard spatial models will exhibit finite-sample bias and inaccurate variance estimation \citep{paciorek2010importance}. } For the theoretical results, we {assume} a random spatial design, but our empirical results verify that the proposed method can perform well if the $\bS_i$ are defined deterministically. Assume $\bY$ and the columns of $\bA$ are centered {with mean} 0.

\subsection{Double Spatial Regression estimator}\label{s:practical_dsr}


Double Spatial Regression estimates the spatial trends $g_0, m_{01}, ..., m_{0\ell}$ and covariate effects $\bZ^T\boldsymbol{\theta}_{00}$, $\bZ^T\boldsymbol{\theta}_{01}$, ..., $\bZ^T\boldsymbol{\theta}_{0\ell}$, subtracts those from $\bY, \bA_{\cdot 1}, ..., \bA_{\cdot \ell}$, and then uses these residuals to estimate $\bbeta_0$. The estimates of $g_0, m_{01}, ..., m_{0\ell}$, and $\bZ^T\boldsymbol{\theta}_{00}$, $\bZ^T\boldsymbol{\theta}_{01}$, ..., $\bZ^T\boldsymbol{\theta}_{0\ell}$, are obtained by cross-fitting: for example, to estimate the values of $g_0$ for the data in fold $k$, an estimate $\widehat{g}_0$ of the function $g_0$ is obtained using the data not in fold $k$, and $\widehat{g}_0$ is then evaluated on the data in fold $k$.

For cross-fitting, the full data set is randomly (we use a completely random sample) partitioned into $K$ folds each of size $\frac{n}{K}$, assumed to be an integer. Let $f_i \in \{1, ..., K\}$ denote the fold assignment of observation $i$. For $k = 1, ..., K$, let $\bk = \{ i : f_i = k\}$ denote the set of indices assigned to fold $k$ and $\bk^C = \{ i : f_i \neq k\}$ denote the set of indices assigned to the complement of fold $k$. Then, for example, $\bY_{\bk}$ is the vector of elements of $\bY$ assigned to the $k$th fold and $\bY_{\bk^C}$ is the vector of elements of $\bY$ not assigned to the $k$th fold. {For a generic matrix $\bX \in \mathbb{R}^{n \times p}$, the $i$th row of $\bX$ is denoted $\bX_i$, the submatrix consisting of those rows of $\bX$ whose indices are in $\bk$ is denoted $\bX_{\bk}$ (and likewise for $\bk^C$), the $j$th column is denoted $\bX_{\cdot j}$, and the $j$th column vector of the submatrix corresponding to the $k$th fold is denoted $\bX_{\bk,j}$.}  Denote the full data as $\bW$, the data in fold $k$ as $\bW_{\bk}$, and the data in the complement of fold $k$ as $\bW_{\bk^C}$.

Universal Kriging using the Mat\`ern correlation function is used to estimate $g_0, m_{01}, ..., m_{0\ell}$ \citep{stein1999interpolation}.  
The Mat\`ern correlation function is:
\begin{align}\label{e:matern}
    C_{\gamma, \tau}(\bS_i, \bS_j) := \frac{2^{1-\tau}}{\Gamma(\tau)}\biggl(\sqrt{2\tau}\frac{\|\bS_i - \bS_j\|_2}{\gamma}\biggr)^{\tau} K_{\tau} \biggl(\sqrt{2\tau}\frac{\|\bS_i - \bS_j\|_2}{\gamma} \biggr),
\end{align}
where $K_{\tau}$ is the modified Bessel function of the second kind. The Mat\`ern correlation function has two parameters: $\gamma$ controls the range of spatial dependence and $\tau$ controls the smoothness (mean-square differentiability) of the process. A variance parameter $\omega^2$ is typically defined so that $\omega^2 C(\bA, \bB)$ is a covariance matrix.


For each fold $k$ and for each of $\bY, \bA_{\cdot 1}, ..., \bA_{\cdot \ell}$, Models (\ref{eq:SDML_practical_model1}) and (\ref{eq:SDML_practical_model2}) are fit to $\bW_{\bk^C}$, using a spatial random effect with Mat\`ern covariance to model the spatial terms $g_0, m_{01}, ..., m_{0\ell}$, obtaining preliminary parameter estimates for each model fit. The preliminary estimates obtained on fold $k$ for $\omega$, the correlation parameters used in (\ref{e:matern}), and the slope parameters used in (\ref{eq:SDML_practical_model1}) and (\ref{eq:SDML_practical_model2}), are denoted by a tilde and an additional subscript $k$; note that $\widetilde{\bbeta_0}$, a preliminary estimate of $\bbeta_0$, is included. 


The universal Kriging equations used to estimate the spatial trends on fold $k$ are:
\begin{align}\label{eq:univ_kriging_eqn1}
\begin{split}
    \widehat{g}_0(\bS_{\bk}) &= \widetilde{\omega}_{0k}C_{\widetilde{\gamma}_{0k},\widetilde{\tau}_{0k}}(\bS_{\bk}, \bS_{\bk^C})\biggl(\widetilde{\omega}_{0k}C_{\widetilde{\gamma}_{0k},
    \widetilde{\tau}_{0k}}(\bS_{\bk^C}, \bS_{\bk^C}) + \widetilde{\sigma}_{0k}^2 \bI\biggr)^{-1}\\
    & \quad \quad \quad \times(\bY_{\bk^C} - \bA_{\bk^C}^T\widetilde{\bbeta}_{0k} - \bZ_{\bk^C}^T \widetilde{\boldsymbol{\theta}}_{00k})
    \end{split}\\
    \widehat{m}_{0j}(\bS_{\bk}) &=  \widetilde{\omega}_{jk}C_{\widetilde{\gamma}_{jk},\widetilde{\tau}_{jk}}(\bS_{\bk}, \bS_{\bk^C})\biggl(\widetilde{\omega}_{jk}C_{\widetilde{\gamma}_{jk},\widetilde{\tau}_{jk}}(\bS_{\bk^C}, \bS_{\bk^C}) + \widetilde{\sigma}_{jk}^2 \bI\biggr)^{-1}(\bA_{\bk^C, j} -\bZ_{\bk^C}^T \widetilde{\boldsymbol{\theta}}_{0jk}),  \label{eq:univ_kriging_eqn2}
\end{align}
for $j=1,...,p$. For this working prediction model we also assume homoskedasticity of all error terms, with variance parameters  $\sigma^2_0, \sigma^2_1, ..., \sigma^2_\ell$. The estimated covariate effects for fold $k$ are $\bZ_{\bk}^T\widetilde{\boldsymbol{\theta}}_{00k}, \bZ_{\bk}^T\widetilde{\boldsymbol{\theta}}_{01k}, ..., \bZ_{\bk}^T\widetilde{\boldsymbol{\theta}}_{0\ell k}$.

{To obtain universal Kriging estimates, we used the \texttt{GpGp} R package \citep{GpGp}, which estimates all covariance and slope parameters using restricted maximum likelihood (REML), and scales well to large samples by using the Vecchia approximation \citep{vecchia1988estimation}. 
}
Estimates are combined across folds to obtain $\widehat{g}_0(\bS)$ and $\widehat{m}_{01}(\bS), ..., \widehat{m}_{0 \ell}(\bS)$, and $\widehat{\bZ^T\boldsymbol{\theta}_{00}}, \widehat{\bZ^T\boldsymbol{\theta}_{01}}, ..., \widehat{\bZ^T\boldsymbol{\theta}_{0 \ell}}$. Denote the DSR estimator by $\widehat{\boldsymbol{\beta}}_{DSR}$.
Letting $\widehat{\bV}_{\cdot j} = \bA_{\cdot j} -\widehat{\bZ^T\boldsymbol{\theta}_{0j}}- \widehat{m}_{0j}(\bS)$, and $\widehat{\bU} = \bY - \bA^T\widehat{\bbeta}_{DSR} - \widehat{\bZ^T\boldsymbol{\theta}_{00}} - \widehat{g}_0(\bS)$:
\begin{align}
    \widehat{\bbeta}_{DSR} &= (\widehat{\bV}^T\bA)^{-1}\widehat{\bV}^T(\bY - \widehat{\bZ^T {\boldsymbol{\theta}}_{00}} - \widehat{g}_0(\bS))\label{e:dsr_practical}\\
    \widehat{Var}(\widehat{\bbeta}_{DSR}) &= (\widehat{\bV}^T\bA)^{-1} \left(\sum_{i=1}^n \widehat{U}_i^2\widehat{\bV}_i \widehat{\bV}_i^T\right)\left((\widehat{\bV}^T\bA)^{-1}\right)^T,\label{e:dsr_practical_var}
\end{align}
which are the estimators provided in \cite{dml}. The DSR algorithm is given by Algorithm \ref{a:SDML_practical}, and a version without cross-fitting, more closely corresponding to \cite{andrews1994asymptotics}, is provided in Supplementary Materials Section S4.

\begin{algorithm}[H]\caption{DSR estimation of $\bbeta_0$ with cross-fitting}\label{a:SDML_practical}
\begin{algorithmic}
    \Require Response vector $\bY \in \mathbb{R}^{n}$, location matrix $\bS \in \mathbb{R}^{n \times 2}$, treatment matrix $\bA \in \mathbb{R}^{n \times \ell}$, covariate matrix $\bZ \in \mathbb{R}^{n \times m}$.
    \Ensure Estimate $\widehat{\bbeta}_{DSR}$ of $\bbeta_0 \in \mathbb{R}^{\ell}$ and estimate $\widehat{Var}(\widehat{\bbeta}_{DSR})$ of $Var(\widehat{\bbeta}_{DSR}) \in \mathbb{R}^{\ell \times \ell}$
    \State Partition the data into $K$ random folds so that the size of each fold is $\frac{n}{K}$.
    \For{$k=1, ..., K$}
    \State Using $\bW_{\bk^C}$, obtain $\widetilde{\boldsymbol{\beta}}_{0k},$ and for $j=0, ..., \ell$, $\widetilde{\boldsymbol{\theta}}_{0jk}$,  $\widetilde{\gamma}_{jk}, \widetilde{\tau}_{jk}, \widetilde{\sigma}_{jk}, \widetilde{\omega}_{jk}$ by fitting Models (\ref{eq:SDML_practical_model1}) and (\ref{eq:SDML_practical_model2}) using \texttt{GpGp}.
    \State Obtain $\widehat{g}_0(\bS_{\bk})$ by (\ref{eq:univ_kriging_eqn1}) (approximated by \texttt{GpGp}).

    \For{$j=1, ..., \ell$}
    \State Obtain $\widehat{m}_{0j}(\bS_{\bk})$ by (\ref{eq:univ_kriging_eqn2}) (approximated by \texttt{GpGp}).
    \State $\widehat{\bA}_{\bk, j} \leftarrow \bZ_{\bk}^T\widetilde{\boldsymbol{\theta}}_{0jk} + \widehat{m}_{0j}(\bS_{\bk}) $

    \EndFor
    \EndFor
    
    \State Combine the $\widehat{g}_0(\bS_{\bk})$, $\bZ_{\bk}^T \widetilde{\boldsymbol{\theta}}_{00k}$, $\widehat{\bA}_{\bk}$ across folds to obtain $\widehat{g}_0(\bS)$, $\widehat{\bZ^T\boldsymbol{\theta}_{00}}$, and $\widehat{\bA}$ respectively.
    \State $\widehat{\bV} \leftarrow \bA - \widehat{\bA}$
    \State $\widehat{\bbeta}_{DSR} \leftarrow (\widehat{\bV}^T\bA)^{-1}\widehat{\bV}^T(\bY - \widehat{\bZ^T {\boldsymbol{\theta}}_{00}} - \widehat{g}_0(\bS))$
    \State $\widehat{\bU} = \bY - \bA^T\widehat{\bbeta}_{DSR} - \widehat{\bZ^T\boldsymbol{\theta}_{00}} - \widehat{g}_0(\bS)$
     \State $\widehat{Var}(\widehat{\bbeta}_{DSR}) \leftarrow (\widehat{\bV}^T\bA)^{-1} \left(\sum_{i=1}^n \widehat{U}_i^2\widehat{\bV}_i \widehat{\bV}_i^T\right)\left((\widehat{\bV}^T\bA)^{-1}\right)^T$
    \State Return $\widehat{\bbeta}_{DSR}, \widehat{Var}(\widehat{\bbeta}_{DSR})$
\end{algorithmic}
\end{algorithm}


If cross-fitting is used, different random allocations of observations into folds will produce different estimates of $\bbeta_0$. For a scalar $\beta_0$, \cite{fuhr2024estimating} recommends using enough folds that the estimates are relatively stable, and pick the median point and variance estimates out of as large a number of estimates as is feasible. The marginal median could be selected for each component of $\widehat{\bbeta}_{DSR}$ if it has length greater than 1.

In Supplementary Material Section S2, we prove root-$n$ consistency and asymptotic normality for a different DSR estimator, also based on \cite{dml}, which more closely resembles gSEM in that in the first stage, $\bY$ is regressed directly onto $\bS$, ignoring treatment variables. Covariates are also considered treatment variables, a limitation of our theory. We term this the ``theoretical DSR" estimator and denote it $\widehat{\bbeta}_{TDSR}$. Instead of Mat\`ern covariance, $\widehat{\bbeta}_{TDSR}$ estimates latent functions of space using Kriging with squared exponential covariance, with covariance parameter selection by a training-validation split \citep{eberts2013optimal}, which allows more thorough theoretical analysis. In summary, for $\mathcal{S} \subset \mathbb{R}^2$, the regularity conditions for $\widehat{\bbeta}_{TDSR}$ require that the unknown functions of space are smoother than once-differentiable functions, and reside in $L_2(\mathbb{R}^2)\cap L_{\infty}(\mathbb{R}^2)$, as well as other conditions. Under Assumptions A1-A6 in Supplemental Materials Section S2, we show root-$n$ consistency and asymptotic normality of $\widehat{\bbeta}_{TDSR}$:
 
\begin{theorem}\label{t:1}
If Assumptions A1 -- A6 are met, and $\widehat{\bbeta}_{TDSR}$ and $\widehat{Var}(\widehat{\bbeta}_{TDSR})$ are obtained by Algorithm S2, then
\begin{align*}
 \sqrt{n}\boldsymbol{\Sigma}^{-1/2}(\widehat{\bbeta}_{TDSR}- \bbeta_{0}) \xrightarrow[]{D} N(\mathbf{0}, \mathbf{I}_p), \text{and} \\
 \widehat{Var}(\widehat{\bbeta}_{TDSR})^{-1/2}(\widehat{\bbeta}_{TDSR} - \bbeta_0) \xrightarrow[]{D} N(\mathbf{0},\mathbf{I}_p),  
\end{align*}
where $\boldsymbol{\Sigma}=E[ \bV_i\bV_i^T]^{-1}E[U_i^2\bV_i\bV_i^T](E[ \bV_i\bV_i^T]^{-1})$ is the approximate variance of $\widehat{\bbeta}_{TDSR}$.
\end{theorem}

Theorem \ref{t:1} does not apply directly to the estimator $\widehat{\bbeta}_{DSR}$. Root-$n$ consistency of $\widehat{\bbeta}_{DSR}$ and asymptotic normality of $\widehat{Var}(\widehat{\bbeta}_{DSR})^{-1/2}(\widehat{\bbeta}_{DSR} - \bbeta_0)$ require convergence of the estimates from (\ref{eq:univ_kriging_eqn1}) and (\ref{eq:univ_kriging_eqn2}) at rates of $o_P(n^{-1/4})$, along with correct model specification and further standard regularity conditions \citep{dml}. \cite{andrews1994asymptotics} provides sufficient conditions to forego cross-fitting, namely that both $(g_0, m_{01},...,m_{0\ell})$ and $(\widehat{g}_0, \widehat{m}_{01},...,\widehat{m}_{0\ell})$ obey additional smoothness conditions. These required results for GP regression with Mat\`ern correlation and hyperparameters estimated via REML are not established to our knowledge, so use of $\widehat{\bbeta}_{DSR}$ is theoretically justified, and cross-fitting can be skipped if desired, with assumptions about the convergence rates of $\widehat{g}_0, \widehat{m}_{01},...,\widehat{m}_{0\ell}$, and smoothness of both $(g_0, m_{01},...,m_{0\ell})$ and $(\widehat{g}_0, \widehat{m}_{01},...,\widehat{m}_{0\ell})$. Although equally thorough theoretical analysis is not possible, in simulations, $\widehat{\bbeta}_{DSR}$ performed much better than $\widehat{\bbeta}_{TDSR}$.

\section{Simulation Study}\label{s:simulation}
\subsection{Simulation study outline}\label{s:sim_outline}

In this simulation study, we evaluate several estimators' performance in finite samples. We include the DSR estimator presented in Section \ref{s:practical_dsr} (Equations \ref{e:dsr_practical} and \ref{e:dsr_practical_var}), which we refer to as ``DSR" in this section, and the ``theoretical DSR" estimator studied in Section S2 (Algorithm S2). Cross-fitting is used for all DSR results presented in this section; the Supplementary Materials Section S5 include results for both versions of DSR without cross-fitting. DSR estimators with cross-fitting used $K=5$ folds. 
{Different estimators of spatial trends, use of cross-fitting, and types of DSR estimators were used to isolate characteristics important to good empirical performance, and are detailed in the Supplemental Materials Section S4.} Comparison methods were the geoadditive structural equation model (gSEM) of \cite{thaden2018structural}, implemented using splines for geostatistical data (Supplementary Materials Algorithm S1), Spatial+ of \cite{dupont2022spatial+}, the naive (unadjusted for spatial confounding) spatial linear mixed model (LMM), spline regression using generalized cross-validation (GCV) and REML to estimate the smoothing parameter for a smooth function of space \citep{wood2017generalized}, and ordinary least squares (OLS). The nonparametric shift estimator studied in \cite{gilbert} was tried but was not able to obtain estimates in most scenarios we used.
Spatial+ removes a fitted spatial surface from the treatment variable, {and then performs regression of the outcome onto the residuals with a thin-plate spline smooth over space}. The naive spatial linear mixed model (LMM) is fitted using the R package \texttt{GpGp} \citep{GpGp} and spline models fit using the R package \texttt{mgcv} \citep{wood2011gam}. Simulations were carried out in R version 4.2.1 \citep{rcore}.

\subsection{Data generating process}\label{s:sim_data}

For our main simulations, we generate data from Gaussian processes in order to control the smoothness of the generated processes. 
Observations were drawn from multivariate normal distributions with selected covariances, equivalent to drawing a random function from a Gaussian process prior with that kernel.
The data-generating process for the main simulations is similar to that used in \cite{marques_explicitly_correlated_grf}:
\begin{align*}
    \bY &= \beta_0 \bA + \bZ + \boldsymbol{\epsilon}, \quad \epsilon_i \stackrel{i.i.d.}{\sim} N(0, \sigma^2_Y), \\
    \begin{bmatrix} \bA \\  \bZ \end{bmatrix} &\sim N \begin{pmatrix} \begin{bmatrix} \mathbf{0} \\ \mathbf{0} \end{bmatrix}, \begin{bmatrix} 
    \boldsymbol{\Sigma}_{\bA} +\sigma^2_A\bI & \rho \boldsymbol{\Sigma}_{\bA}^{1/2}(\boldsymbol{\Sigma}_{\bZ}^{1/2})^T \\
    \rho \boldsymbol{\Sigma}_{\bZ}^{1/2}(\boldsymbol{\Sigma}_{\bA}^{1/2})^T & \boldsymbol{\Sigma}_{\bZ}
    \end{bmatrix}\end{pmatrix},\label{e:simAU}
\end{align*}
where $\bY, \bA, \bZ \in \mathbb{R}^n$, $\boldsymbol{\Sigma}_{\bA}$ and $\boldsymbol{\Sigma}_{\bZ}$ are spatial correlation matrices in $\mathbb{R}^{n \times n}$, and $\rho \in [-1, 1]$. $\bA$ is the independent variable, $\beta_0$ is the parameter of interest, and $\bZ$ is an unobserved spatial confounder. {Since the main regularity conditions presented in Section S2 relate to the differentiability of the unknown functions of space, we used two spatial covariance matrices for both $\boldsymbol{\Sigma}_{\bA}$ and $\boldsymbol{\Sigma}_{\bZ}$ that generate realizations of smooth and rough, i.e., more and less differentiable,} functions:

\begin{itemize}
    \item Smooth: Gneiting covariance function \citep{schlather2015analysis} with range parameter $\gamma=0.2$, and variance 1
    \item Rough: Mat\`ern covariance function with smoothness parameter $\tau=1.5$, range parameter {$\gamma=0.072$}, and variance 1
\end{itemize}
The Gneiting covariance function approximates a {squared exponential} covariance function but is less prone to producing computationally singular covariance matrices \citep{schlather2015analysis}. These were both used as implemented in the R package \texttt{geoR} \citep{geoR}. {Range parameters were chosen so that the distance at which spatial correlation is equal to 0.05 were approximately equal. Example draws of spatial surfaces from these distributions are visualized in the Supplemental Materials section S4.}

For our main simulations, $n=1000$, $\rho=0.5$, $\sigma^2_A=0.1^2$, $\sigma^2_Y=1$, and $\beta_0=0.5$. The relative lack of i.i.d. variance in $\bA$ provides less unconfounded variation in the treatment, causing severe spatial confounding bias in finite samples. Spatial locations were randomly sampled uniformly from $[0, 1]^2$. Several additional simulated scenarios generated data subject to heteroskedasticity, nonstationarity, non-Gaussianity of the confounder $Z_i$ or noise $\epsilon_i$, or considering deterministically-defined latent functions of space $g_0$ and $m_0$  rather than random functions drawn from GP priors. These scenarios are discussed in the Supplementary Materials Section S4. 

Four hundred Monte Carlo replications were performed for each main simulation scenario. Confidence intervals for Spatial+ and gSEM were obtained using 100 nonparametric bootstrap replicates. Methods were compared using their bias, standard error, mean squared error (MSE), coverage, and confidence interval length.

\subsection{Simulation results}\label{s:sim_results}

Figure \ref{fig:sim_pointests} displays sampling distributions of $\widehat{\beta}_0 - \beta_0$ for several scenarios and a subset of methods. Table \ref{tab:sim_CI} displays corresponding coverage and mean confidence interval lengths. DSR has the lowest bias and shortest confidence intervals of the methods with at least near-nominal coverage. Theoretical DSR performs fairly poorly in terms of bias, coverage, and confidence interval length. gSEM and Spatial+ are improvements over the linear mixed model, and achieve nominal coverage, but they have higher bias and wider confidence intervals than DSR. These scenarios were chosen to be particularly challenging cases of spatial confounding, {in that there is little unconfounded variation in the treatment variable,} so the linear mixed model performs poorly by all metrics.

\begin{figure}
    \centering
    \includegraphics[scale=.35]{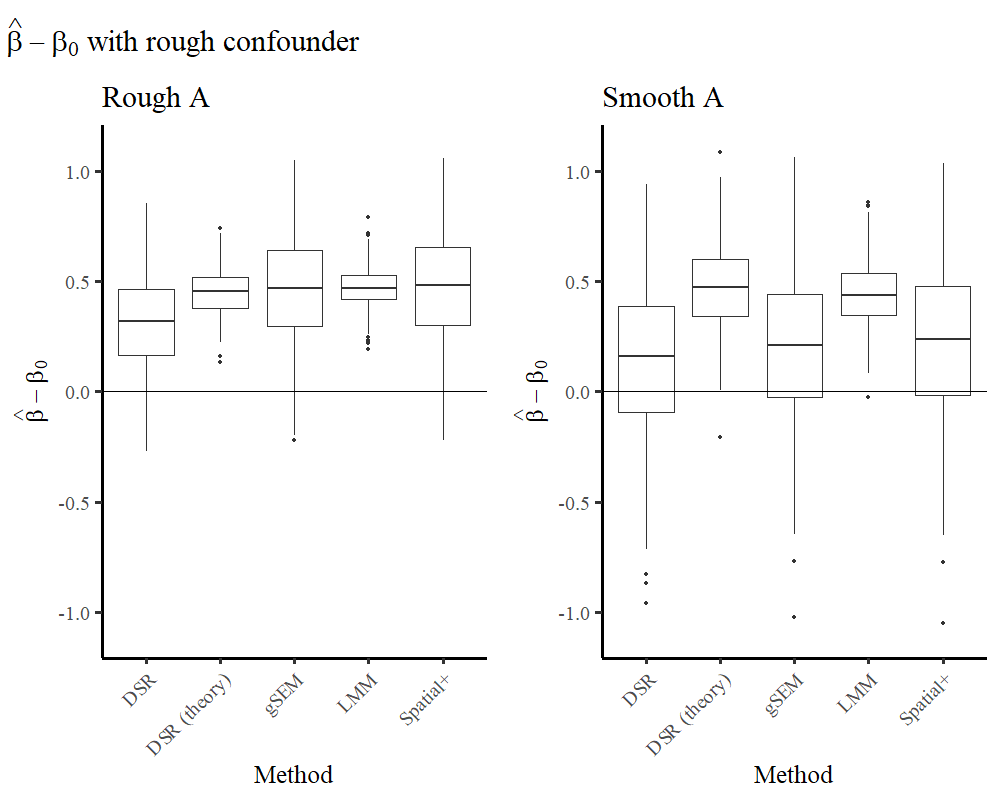}
    \includegraphics[scale=.35]{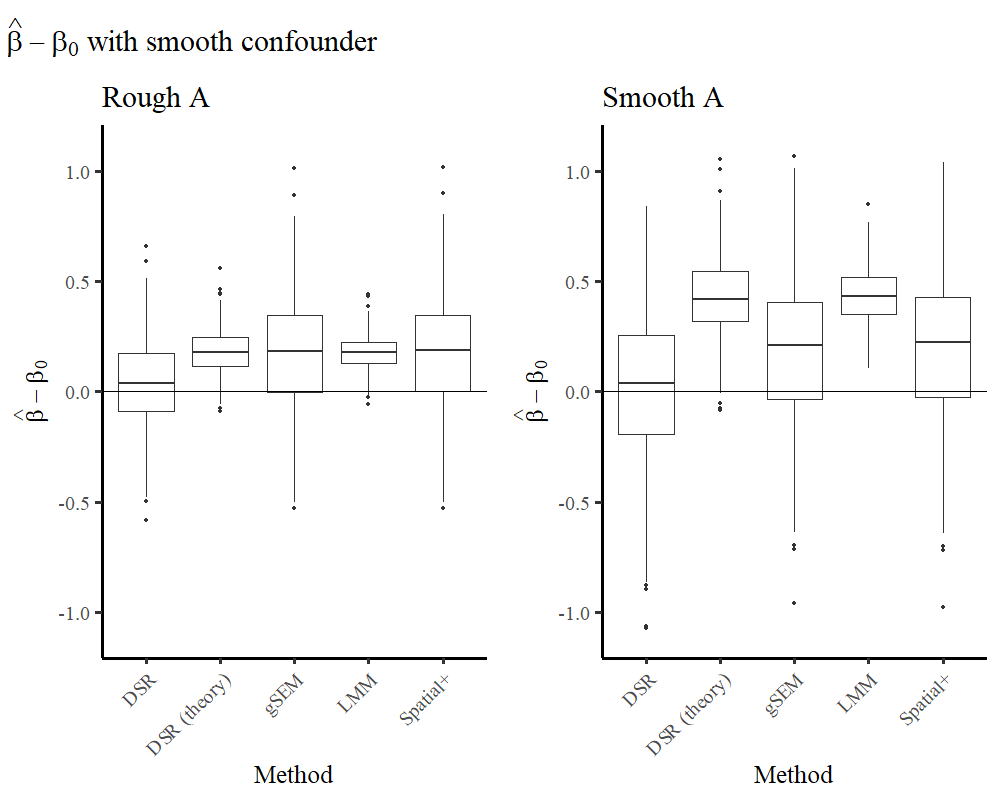}
    \caption{Sampling distribution of $\widehat{\beta} - \beta_0$ from DSR, theoretical DSR, gSEM, LMM, and Spatial+. On the left side, the unobserved confounder $Z$ is smooth, while on the right side it is rough. The smoothness of the treatment variable $A$ is similarly varied.}
    \label{fig:sim_pointests}
\end{figure}

\begin{table}[h!]
    \centering
    \caption{Coverage (mean confidence interval length) of 95\% confidence intervals for DSR, theoretical DSR, gSEM, LMM, and Spatial+ in illustrative scenarios. Scenarios varied by the smoothness of the treatment variable $A$ and the unobserved spatial confounder $Z$.}
    \begin{tabular}{ccccccc}
         $A$&$Z$&DSR & Theoretical DSR & gSEM  & LMM & Spatial+  \\
         \hline
         Rough& Rough  & 0.63 (0.79) & 0.01 (0.43) & 0.70 (1.25) & 0.00 (0.33) & 0.70 (1.26) \\
         Smooth & Rough  & 0.93 (1.45) & 0.37 (0.82) & 0.95 (1.56) & 0.09 (0.50) & 0.95 (1.63) \\
         Rough & Smooth  & 0.97 (0.78) & 0.61 (0.41) & 0.97 (1.24) & 0.34 (0.30) & 0.97 (1.24)\\
         Smooth & Smooth  & 0.96 (1.46) & 0.45 (0.79) & 0.95 (1.55) & 0.07 (0.46) & 0.95 (1.62) \\
    \end{tabular}
    
    \label{tab:sim_CI}
\end{table}

Full simulation result tables are in the Supplementary Material Section S5, Tables S1-S13. 
In many scenarios, especially the deterministic scenarios where the latent functions of space were predefined, using cross-fitting 
{ proved beneficial}.
DSR implemented with spline regression without cross-fitting outperformed gSEM in many scenarios, indicating that direct estimation of $g_0(\bS)$ rather than {ignoring $\bA$ in the first stage} generally improves performance. 
{Replacing the GP estimator in theoretical DSR by Kriging with Mat\`ern covariance and REML estimation greatly improved performance, but in a number of scenarios resulted in slight undercoverage, and sometimes worse bias than DSR. To study the robustness of DSR to the choice of correlation function, we applied DSR with Mat\`ern covariance smoothness fixed at 4.5. This resulted in slightly better estimates generally, but this may be because many of our confounding scenarios were generated with smooth covariance functions.}

{In more difficult scenarios, such as those with rougher unobserved confounders, and the deterministic scenarios, DSR's performance suffered but still outperformed other methods in bias, with slightly worse coverage than gSEM and Spatial+.}
The confounding scenario in which $\bA$ and $\bZ$ were generated with exponential covariance proved too difficult for any method to perform well and all suffered from substantial bias (Table S10 in the Supplementary Materials). When $\sigma^2_A=1$, DSR had higher bias than gSEM and the shift estimand, although all methods had low bias in this scenario.

\section{PFAS Data Analysis}\label{s:analysis}

Per- and polyfluoroalkyl substances (PFAS) are synthetic chemicals resistant to water and stains used in consumer products and manufacture of chemicals, many of which are frequently detected in humans due to resistance to degradation and solubility in water \citep{national2022guidance}. {PFAS have been associated with thyroid disruption in humans, although results have been inconsistent \citep{national2022guidance, coperchini2021thyroid}. Animal models indicate that PFAS can impact thyroid hormone function through a variety of mechanisms including thyroid hormone synthesis \citep{bali2024per}. Because PFAS reflect a broad class of chemicals, researchers are using in silico modeling to try to characterize these impacts \citep{baralic2024exploring}.}

In \cite{sun2016legacy}, high concentrations of a number of PFAS were detected in water at several collection points of the Cape Fear River in North Carolina (NC), the drinking water source for over a million people. One source of contamination is a fluorochemical manufacturing plant slightly downstream of Fayetteville, North Carolina \citep{kotlarz2020measurement, kotlarz2023fay}. In 2017, the GenX Exposure Study began collecting data on North Carolina residents in drinking water exposed communities in order to understand their patterns of exposure to PFAS and possible health effects. The data used in this analysis comes from blood samples collected from study participants in the private well community near the chemical plant in October 2021. 
We study associations between five PFAS, because each was detected in at least 80\% of our samples, and thyroid stimulating hormone (TSH). 

The sample from this community is comprised of 98 adult women who obtain their water from private wells, who lived near the Fayetteville Works chemical plant, and who did not have thyroid disease at the time of data collection. Women with thyroid disease were excluded because thyroid medication can stabilize thyroid hormone levels. The outcome variable we studied was blood serum concentration of TSH. We focus on this outcome and community because the geographic area was relatively small (about 250 square miles), and the outcome and some exposures showed spatial dependence, so this analysis is a candidate for a spatial confounding adjustment. For example, other unmeasured pollutants might affect health outcomes and be correlated with the PFAS used in the analysis due to similar sources or spread. TSH was {natural} log-transformed to obtain approximately normal residuals.

Exposure variables were serum concentrations of five different PFAS detected in more than 80\% of our sample: {perfluoroheptanesulfonic acid (PFHpS), perfluorooctanoic acid (PFOA), perfluorononanoic acid (PFNA), perfluorooctane sulfonate (PFOS), and perfluorohexanesulfonic acid (PFHxS)}. The blood serum concentrations of PFAS and thyroid hormones were measured in the same samples. \cite{wallis2023estimation} summarizes results from \cite{li2018half}, \cite{li2022determinants}, and \cite{zhang2013biomonitoring} that show fairly long half-lives of these chemicals, indicating that the measurements taken for these PFAS are good proxies of cumulative exposure. We controlled for several confounders that could be associated with both exposures and outcome: age, sex, race, smoking status and current alcohol consumption. Table \ref{tab:summary_cont} provides summary statistics of the continuous variables. {The spatial proportion of residual variance is the variance of spatially-correlated residuals, divided by the total residual variance, in a LMM regressing that variable onto covariates.} Thirty-nine members (40\%) of the sample had ever smoked, forty-two members (43\%) currently drank, and 78 (80\%) were White, ten (10\%) were Black, and ten (10\%) were other or multiple races. The median age was 59 years, with a mean of 56.9 and standard deviation of 14.7.


Double spatial regression may offer an improved ability to analyze the effects of spatially-correlated chemical exposures, which may otherwise be difficult to estimate without adjustment for spatial confounding in standard models. {DSR's regularity conditions, while informative as to the general strength of the method, are not possible to check in practice, as they are assumptions on unknown functions of space. We recommend simply comparing model estimates from DSR and the LMM, and possibly OLS. A substantial difference between DSR and LMM results may indicate that spatial confounding was present in the LMM and adjusted for more successfully in DSR's estimates.} 


We estimated the joint associations between the five PFAS and TSH using non-spatial linear regression {(OLS)}, a spatial linear mixed model (LMM) estimated with \texttt{GpGp} using Mat\`ern covariance,  gSEM, Spatial+, and the DSR estimator with cross-fitting in Section \ref{s:practical_dsr}. A large number of folds, 45, was used in the DSR estimator to improve stability of estimates across random splits; the marginal median estimates, and corresponding variance estimates, were chosen from 11 runs with different random sample splits. Linear regression diagnostics showed that an assumption of a linear association of all PFAS serum concentrations with {log(TSH)} was reasonable.

\begin{table}
    \centering
    \caption{Summary statistics of treatment and response variables used in the analysis of PFAS exposure's association with TSH in 98 women from the private well community near the Fayetteville Works fluorochemical plant. {The spatial proportion of residual variance is the variance of spatially-correlated residuals, divided by the total residual variance, in a spatial linear mixed model regressing that variable onto covariates. TSH was log-transformed before estimating the spatial proportion of residual variance.} }
    \begin{tabular}{ccccc}
         Variable&  Median&  Mean&  SD & Spatial prop. of resid. variance\\
         PFHpS (ng/mL)&  0.31&  0.47&  0.60 & 0.71\\
         PFHxS (ng/mL)&  1.96&  2.34&  1.58 & 0.00\\
         PFNA  (ng/mL)&0.55&  0.68&  0.54 & 0.33\\
         PFOS  (ng/mL)&5.87&  7.02&  5.16 & 0.17\\
         PFOA  (ng/mL)&2.17&  2.50&  1.79 & 0.46\\
 TSH (uIU/mL)& 1.34& 1.51& 0.71 & 0.41\\
    \end{tabular}
    
    \label{tab:summary_cont}
\end{table}

{Regression analysis results, shown in Figure \ref{fig:forest_plot}, indicate increasing magnitude and statistical significance of the $\beta$ parameter for PFHpS, from OLS to LMM to DSR. This trend of increasing magnitude and significance with increasing spatial confounding adjustment suggests the presence of spatial confounding with respect to PFHpS. About 70\% of the residual variance in PFHpS was attributable to spatial variation, making estimation of its regression coefficient more susceptible to spatial confounding in this analysis.}

\begin{figure}[h!]
    \centering
    \includegraphics[width=0.80\linewidth]{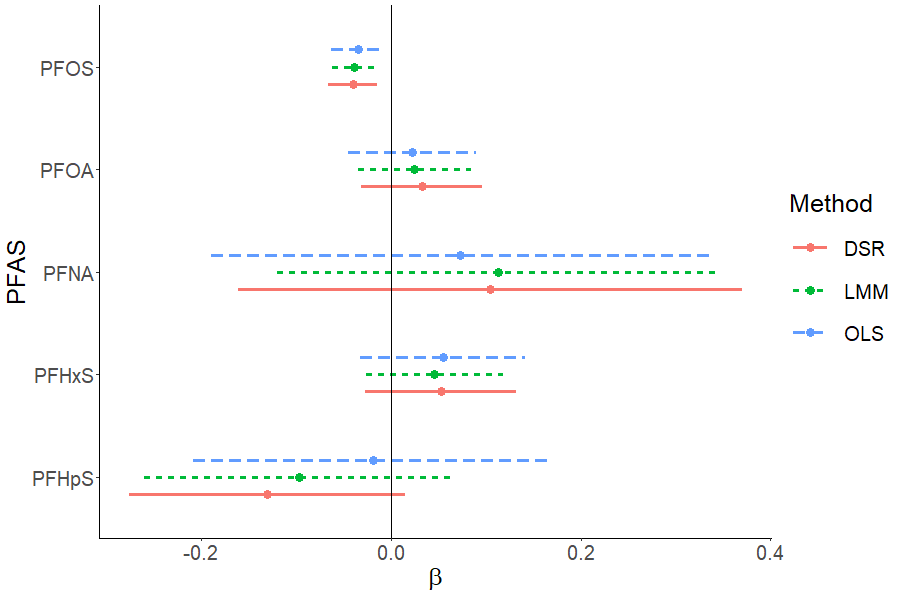}
    \caption{Forest plot of point estimates and 95\% confidence intervals, based on the limiting normal distribution, for regression coefficients $\beta$ for each PFAS. Results are presented for each method: ordinary least squares (OLS), linear mixed model (LMM), and double spatial regression (DSR). Of the five PFAS, PFHpS had the strongest spatial autocorrelation, with approximately 70\% of its residual variance explainable by a spatial random effect.}
    \label{fig:forest_plot}
\end{figure}

\section{Discussion}


Double spatial regression (DSR) is a method to estimate linear regression coefficients using point-referenced data that may be spatially confounded. In simulated severe confounding scenarios, DSR outperformed  competing methods in bias and coverage. We derived explicit regularity conditions that govern asymptotic performance of a slightly different DSR estimator. These regularity conditions are fairly unrestrictive but this ``theoretical DSR" estimator performed poorly in simulations of severe confounding. DSR also allows closed-form variance estimation, which many competing methods have lacked. Based on simulation results, we recommend using the DSR estimator from Section \ref{s:practical_dsr} with cross-fitting.

We linked the problem of spatial confounding to existing results in the semiparametric regression literature, which has provided insight into why gSEM is able to reduce bias compared to the naive spatial linear mixed model, and studied Double Machine Learning estimators (which we termed Double Spatial Regression in our narrower application) that outperform others proposed so far. These existing results explain when and why these types of two-stage estimators are able to correct for bias, even when the initial nonparametric regression estimates converge to the true functions slowly. 
Simulation results indicate that both use of GP regression using Mat\`ern covariance, and direct estimation of the function $g_0$, as opposed to $h_0$, which marginalizes over the distribution of $A_i$, {improved DSR's estimates} compared to gSEM in many scenarios.

DSR is similar in some ways to the method studied in \cite{gilbert}.
This shift estimator was well-suited to estimating the average treatment effect when treatment effects were heterogeneous. In many practical situations there are multiple treatment variables, and it is reasonable to assume a linear and additive association between treatment(s) and response. In this situation, DSR is able to estimate multiple linear regression coefficients jointly, with their covariance matrix in closed-form.
\cite{gilbert} assumed increasing-domain asymptotics, while we assumed a fixed spatial domain.

\cite{gilbert} also suggested using spatially-independent sample splits for cross-fitting, under an assumption of local spatial covariance. With this type of sample splitting, the conditional mean or density functions of the response and treatment variables, conditional on spatial location, are estimated using points in one subset of the spatial domain, and those functions' realized values estimated on a different subset of the spatial domain such that the points in the prediction set are spatially independent from the points in the training set. It is not clear that these estimates can converge to their true values when the training and prediction sets are spatially independent, since the goal is to learn functions of space; deriving conditions on the covariance functions generating the spatially-dependent data that are sufficient to achieve root-$n$ asymptotic normality and consistency when using spatially-independent splits could verify when this strategy can be used.


Further research could analyze convergence rates of functions estimated with Gaussian processes, when hyperparameters are selected using maximum likelihood or REML rather than training-validation splits; the only work we are aware of doing this, although with noiseless observations, is \cite{karvonen2020maximum}. This would help yield more explicit regularity conditions for the main DSR estimator we proposed (although we also proposed using the Vecchia approximation to the full likelihood). Approaches to areal data should also be investigated. Finally, extensions to other types of semiparametric models, including weighted and nonlinear regression, are discussed in \cite{robinson1988root}, \cite{andrews1994asymptotics}, and \cite{dml}.


\backmatter


\section*{Acknowledgements}
Nate Wiecha is supported by the GenX Exposure study, on United States National Institutes of Health Superfund grant number P42 ES031009. The GenX Exposure Study is supported by research funding from the National Institute of Environmental Health Sciences (1R21ES029353), Center for Human Health and the Environment (CHHE) at NC State University (P30 ES025128), the Center for Environmental and Health Effects of PFAS (P42 ES0310095), and the NC Policy Collaboratory. This work was also supported by National Institutes of Health grants R01ES031651-03 and NIH R01ES031651-01, and National Science Foundation grant DMS2152887. The views expressed in this manuscript are those of the authors and do not necessarily represent the views or policies of  the National Institutes of Health.

\section*{Supplementary Materials}

Web Appendices and Tables referenced in Sections \ref{s:intro}, \ref{s:practical_dsr}, \ref{s:sim_outline}, \ref{s:sim_data}, and \ref{s:sim_results} are available with this paper at the Biometrics website on Oxford Academic. Code for DSR estimators and simulation studies is available on GitHub at \url{https://github.com/nbwiecha/Double-Spatial-Regression}. \vspace*{-8pt}

\section*{Data availability statement}
The data underlying this article cannot be shared publicly, for the privacy of individuals that participated in the study. The data will be shared on reasonable request to the corresponding author.


%
 \bibliographystyle{biom} 
\bibliography{refs.bib}

\begin{thebibliography}{}

\bibitem[\protect\citeauthoryear{Andrews}{Andrews}{1994}]{andrews1994asymptotics}
Andrews, D. W.~K. (1994).
\newblock Asymptotics for semiparametric econometric models via stochastic equicontinuity.
\newblock {\em Econometrica} {\bf 62,} 43--72.

\bibitem[\protect\citeauthoryear{Bali, Martin, Almeida, Saunders, and Wilson}{Bali et~al.}{2024}]{bali2024per}
Bali, S.~K., Martin, R., Almeida, N.~M., Saunders, C., and Wilson, A.~K. (2024).
\newblock {Per-and Polyfluoroalkyl (PFAS) Disruption of Thyroid Hormone Synthesis}.
\newblock {\em ACS Omega} {\bf 9,} 39554--39563.

\bibitem[\protect\citeauthoryear{Barali{\'c}, Petkovski, Pileti{\'c}, Mari{\'c}, Buha~Djordjevic, Antonijevi{\'c}, and {\DJ}uki{\'c}-{\'C}osi{\'c}}{Barali{\'c} et~al.}{2024}]{baralic2024exploring}
Barali{\'c}, K., Petkovski, T., Pileti{\'c}, N., Mari{\'c}, {\DJ}., Buha~Djordjevic, A., Antonijevi{\'c}, B., and {\DJ}uki{\'c}-{\'C}osi{\'c}, D. (2024).
\newblock {Exploring Toxicity of Per-and Polyfluoroalkyl Substances (PFAS) Mixture Through ADMET and Toxicogenomic In Silico Analysis: Molecular Insights}.
\newblock {\em International Journal of Molecular Sciences} {\bf 25,} 12333.

\bibitem[\protect\citeauthoryear{Chernozhukov, Chetverikov, Demirer, Duflo, Hansen, Newey, and Robins}{Chernozhukov et~al.}{2018}]{dml}
Chernozhukov, V., Chetverikov, D., Demirer, M., Duflo, E., Hansen, C., Newey, W., and Robins, J. (2018).
\newblock {Double/debiased machine learning for treatment and structural parameters}.
\newblock {\em The Econometrics Journal} {\bf 21,} C1--C68.

\bibitem[\protect\citeauthoryear{Clayton, Bernardinellli, and Montomoli}{Clayton et~al.}{1993}]{clayton:1993}
Clayton, D.~G., Bernardinellli, L., and Montomoli, C. (1993).
\newblock {Spatial correlation in ecological analysis}.
\newblock {\em International Journal of Epidemiology} {\bf 22,} 1193--1202.

\bibitem[\protect\citeauthoryear{Coperchini, Croce, Ricci, Magri, Rotondi, Imbriani, and Chiovato}{Coperchini et~al.}{2021}]{coperchini2021thyroid}
Coperchini, F., Croce, L., Ricci, G., Magri, F., Rotondi, M., Imbriani, M., and Chiovato, L. (2021).
\newblock {Thyroid disrupting effects of old and new generation PFAS}.
\newblock {\em Frontiers in endocrinology} {\bf 11,} 612320.

\bibitem[\protect\citeauthoryear{Cressie}{Cressie}{1993}]{cressiebook}
Cressie, N.~A. (1993).
\newblock {\em Statistics for Spatial Data}.
\newblock Wiley, 2nd edition.

\bibitem[\protect\citeauthoryear{Dupont, Marques, and Kneib}{Dupont et~al.}{2023}]{dupont2023demystifying}
Dupont, E., Marques, I., and Kneib, T. (2023).
\newblock Demystifying spatial confounding.
\newblock {\em arXiv preprint arXiv:2309.16861v1} .

\bibitem[\protect\citeauthoryear{Dupont, Wood, and Augustin}{Dupont et~al.}{2022}]{dupont2022spatial+}
Dupont, E., Wood, S.~N., and Augustin, N.~H. (2022).
\newblock Spatial+: A novel approach to spatial confounding.
\newblock {\em Biometrics} {\bf 78,} 1279--1290.

\bibitem[\protect\citeauthoryear{Eberts and Steinwart}{Eberts and Steinwart}{2013}]{eberts2013optimal}
Eberts, M. and Steinwart, I. (2013).
\newblock {Optimal regression rates for SVMs using Gaussian kernels}.
\newblock {\em Electronic Journal of Statistics} {\bf 7,} 1 -- 42.

\bibitem[\protect\citeauthoryear{Fuhr, Berens, and Papies}{Fuhr et~al.}{2024}]{fuhr2024estimating}
Fuhr, J., Berens, P., and Papies, D. (2024).
\newblock Estimating causal effects with double machine learning--a method evaluation.
\newblock {\em arXiv preprint arXiv:2403.14385} .

\bibitem[\protect\citeauthoryear{Gilbert, Datta, Casey, and Ogburn}{Gilbert et~al.}{2021}]{gilbert}
Gilbert, B., Datta, A., Casey, J.~A., and Ogburn, E.~L. (2021).
\newblock A causal inference framework for spatial confounding.
\newblock {\em arXiv preprint arXiv:2112.14946} .

\bibitem[\protect\citeauthoryear{Gilbert, Ogburn, and Datta}{Gilbert et~al.}{2024}]{gilbert2023consistency}
Gilbert, B., Ogburn, E.~L., and Datta, A. (2024).
\newblock Consistency of common spatial estimators under spatial confounding.
\newblock {\em Biometrika} page asae070.

\bibitem[\protect\citeauthoryear{Guan, Page, Reich, Ventrucci, and Yang}{Guan et~al.}{2022}]{guan_2020}
Guan, Y., Page, G.~L., Reich, B.~J., Ventrucci, M., and Yang, S. (2022).
\newblock {Spectral adjustment for spatial confounding}.
\newblock {\em Biometrika} {\bf 110,} 699--719.

\bibitem[\protect\citeauthoryear{Guinness}{Guinness}{2018}]{GpGp}
Guinness, J. (2018).
\newblock {Permutation and grouping methods for sharpening Gaussian Process approximations}.
\newblock {\em Technometrics} {\bf 60,} 415--429.
\newblock PMID: 31447491.

\bibitem[\protect\citeauthoryear{Hodges and Reich}{Hodges and Reich}{2010}]{hodges2010adding}
Hodges, J.~S. and Reich, B.~J. (2010).
\newblock Adding spatially-correlated errors can mess up the fixed effect you love.
\newblock {\em The American Statistician} {\bf 64,} 325--334.

\bibitem[\protect\citeauthoryear{Karvonen, Wynne, Tronarp, Oates, and S\"{a}rkk\"{a}}{Karvonen et~al.}{2020}]{karvonen2020maximum}
Karvonen, T., Wynne, G., Tronarp, F., Oates, C., and S\"{a}rkk\"{a}, S. (2020).
\newblock {Maximum likelihood estimation and uncertainty quantification for Gaussian Process approximation of deterministic functions}.
\newblock {\em SIAM/ASA Journal on Uncertainty Quantification} {\bf 8,} 926--958.

\bibitem[\protect\citeauthoryear{Khan and Berrett}{Khan and Berrett}{2024}]{khan_berrett}
Khan, K. and Berrett, C. (2024).
\newblock Re-thinking spatial confounding in spatial linear mixed models.

\bibitem[\protect\citeauthoryear{Kotlarz, Guillette, Critchley, Collier, Lea, McCord, Strynar, Cuffney, Hopkins, Knappe, and Hoppin}{Kotlarz et~al.}{2024}]{kotlarz2023fay}
Kotlarz, N., Guillette, T., Critchley, C., Collier, D., Lea, C.~S., McCord, J., Strynar, M., Cuffney, M., Hopkins, Z.~R., Knappe, D.~R., and Hoppin, J.~A. (2024).
\newblock Per- and polyfluoroalkyl ether acids in well water and blood serum from private well users residing by a fluorochemical facility near {Fayetteville, North Carolina}.
\newblock {\em Journal of Exposure Science \& Environmental Epidemiology} {\bf 34,} 97--107.

\bibitem[\protect\citeauthoryear{Kotlarz, McCord, Collier, Lea, Strynar, Lindstrom, Wilkie, Islam, Matney, Tarte, et~al\mbox{.}}{Kotlarz et~al.}{2020}]{kotlarz2020measurement}
Kotlarz, N., McCord, J., Collier, D., Lea, C.~S., Strynar, M., Lindstrom, A.~B., Wilkie, A.~A., Islam, J.~Y., Matney, K., Tarte, P., et~al. (2020).
\newblock Measurement of novel, drinking water-associated {PFAS} in blood from adults and children in {Wilmington, North Carolina}.
\newblock {\em Environmental Health Perspectives} {\bf 128,} 077005.

\bibitem[\protect\citeauthoryear{Li, Andersson, Xu, Pineda, Nilsson, Lindh, Jakobsson, and Fletcher}{Li et~al.}{2022}]{li2022determinants}
Li, Y., Andersson, A., Xu, Y., Pineda, D., Nilsson, C.~A., Lindh, C.~H., Jakobsson, K., and Fletcher, T. (2022).
\newblock Determinants of serum half-lives for linear and branched perfluoroalkyl substances after long-term high exposure—a study in {Ronneby, Sweden}.
\newblock {\em Environment International} {\bf 163,} 107198.

\bibitem[\protect\citeauthoryear{Li, Fletcher, Mucs, Scott, Lindh, Tallving, and Jakobsson}{Li et~al.}{2018}]{li2018half}
Li, Y., Fletcher, T., Mucs, D., Scott, K., Lindh, C.~H., Tallving, P., and Jakobsson, K. (2018).
\newblock Half-lives of {PFOS, PFHxS and PFOA} after end of exposure to contaminated drinking water.
\newblock {\em Occupational and Environmental Medicine} {\bf 75,} 46--51.

\bibitem[\protect\citeauthoryear{Marques, Kneib, and Klein}{Marques et~al.}{2022}]{marques_explicitly_correlated_grf}
Marques, I., Kneib, T., and Klein, N. (2022).
\newblock {Mitigating spatial confounding by explicitly correlating Gaussian random fields}.
\newblock {\em Environmetrics} {\bf 33,} e2727.

\bibitem[\protect\citeauthoryear{{National Academies of Sciences, Engineering, and Medicine}}{{National Academies of Sciences, Engineering, and Medicine}}{2022}]{national2022guidance}
{National Academies of Sciences, Engineering, and Medicine} (2022).
\newblock {\em Guidance on PFAS Exposure, Testing, and Clinical Follow-Up}.
\newblock The National Academies Press, Washington, DC.

\bibitem[\protect\citeauthoryear{Ogburn, Sofrygin, Díaz, and van~der Laan}{Ogburn et~al.}{2024}]{Ogburn02012024}
Ogburn, E.~L., Sofrygin, O., Díaz, I., and van~der Laan, M.~J. (2024).
\newblock Causal inference for social network data.
\newblock {\em Journal of the American Statistical Association} {\bf 119,} 597--611.

\bibitem[\protect\citeauthoryear{Paciorek}{Paciorek}{2010}]{paciorek2010importance}
Paciorek, C.~J. (2010).
\newblock The importance of scale for spatial-confounding bias and precision of spatial regression estimators.
\newblock {\em Statistical Science} {\bf 25,} 107--125.

\bibitem[\protect\citeauthoryear{{R Core Team}}{{R Core Team}}{2023}]{rcore}
{R Core Team} (2023).
\newblock {\em R: A Language and Environment for Statistical Computing}.
\newblock R Foundation for Statistical Computing, Vienna, Austria.

\bibitem[\protect\citeauthoryear{Reich, Hodges, and Zadnik}{Reich et~al.}{2006}]{reich2006effects}
Reich, B.~J., Hodges, J.~S., and Zadnik, V. (2006).
\newblock Effects of residual smoothing on the posterior of the fixed effects in disease-mapping models.
\newblock {\em Biometrics} {\bf 62,} 1197--1206.

\bibitem[\protect\citeauthoryear{{Ribeiro Jr}, Diggle, Christensen, Schlather, Bivand, and Ripley}{{Ribeiro Jr} et~al.}{2023}]{geoR}
{Ribeiro Jr}, P.~J., Diggle, P., Christensen, O., Schlather, M., Bivand, R., and Ripley, B. (2023).
\newblock {\em geoR: Analysis of Geostatistical Data}.
\newblock R package version 1.9-3.

\bibitem[\protect\citeauthoryear{Rice}{Rice}{1986}]{rice1986convergence}
Rice, J. (1986).
\newblock Convergence rates for partially splined models.
\newblock {\em Statistics \& Probability Letters} {\bf 4,} 203--208.

\bibitem[\protect\citeauthoryear{Robinson}{Robinson}{1988}]{robinson1988root}
Robinson, P.~M. (1988).
\newblock Root-n-consistent semiparametric regression.
\newblock {\em Econometrica} {\bf 56,} 931--954.

\bibitem[\protect\citeauthoryear{Schlather, Malinowski, Menck, Oesting, and Strokorb}{Schlather et~al.}{2015}]{schlather2015analysis}
Schlather, M., Malinowski, A., Menck, P.~J., Oesting, M., and Strokorb, K. (2015).
\newblock Analysis, simulation and prediction of multivariate random fields with package randomfields.
\newblock {\em Journal of Statistical Software} {\bf 63,} 1–25.

\bibitem[\protect\citeauthoryear{Schnell and Papadogeorgou}{Schnell and Papadogeorgou}{2020}]{schnell2019mitigating}
Schnell, P.~M. and Papadogeorgou, G. (2020).
\newblock {Mitigating unobserved spatial confounding when estimating the effect of supermarket access on cardiovascular disease deaths}.
\newblock {\em The Annals of Applied Statistics} {\bf 14,} 2069 -- 2095.

\bibitem[\protect\citeauthoryear{Stein}{Stein}{1999}]{stein1999interpolation}
Stein, M. (1999).
\newblock {\em Interpolation of Spatial Data: Some Theory for Kriging}.
\newblock Springer Series in Statistics. Springer New York.

\bibitem[\protect\citeauthoryear{Sun, Arevalo, Strynar, Lindstrom, Richardson, Kearns, Pickett, Smith, and Knappe}{Sun et~al.}{2016}]{sun2016legacy}
Sun, M., Arevalo, E., Strynar, M., Lindstrom, A., Richardson, M., Kearns, B., Pickett, A., Smith, C., and Knappe, D.~R. (2016).
\newblock Legacy and emerging perfluoroalkyl substances are important drinking water contaminants in the {Cape Fear River Watershed of North Carolina}.
\newblock {\em Environmental Science \& Technology Letters} {\bf 3,} 415--419.

\bibitem[\protect\citeauthoryear{Thaden and Kneib}{Thaden and Kneib}{2018}]{thaden2018structural}
Thaden, H. and Kneib, T. (2018).
\newblock Structural equation models for dealing with spatial confounding.
\newblock {\em The American Statistician} {\bf 72,} 239--252.

\bibitem[\protect\citeauthoryear{Vecchia}{Vecchia}{1988}]{vecchia1988estimation}
Vecchia, A.~V. (1988).
\newblock Estimation and model identification for continuous spatial processes.
\newblock {\em Journal of the Royal Statistical Society. Series B (Methodological)} {\bf 50,} 297--312.

\bibitem[\protect\citeauthoryear{Wallis, Kotlarz, Knappe, Collier, Lea, Reif, McCord, Strynar, DeWitt, and Hoppin}{Wallis et~al.}{2023}]{wallis2023estimation}
Wallis, D.~J., Kotlarz, N., Knappe, D.~R., Collier, D.~N., Lea, C.~S., Reif, D., McCord, J., Strynar, M., DeWitt, J.~C., and Hoppin, J.~A. (2023).
\newblock Estimation of the half-lives of recently detected per-and polyfluorinated alkyl ethers in an exposed community.
\newblock {\em Environmental Science \& Technology} {\bf 57,} 15348--15355.

\bibitem[\protect\citeauthoryear{Wood}{Wood}{2011}]{wood2011gam}
Wood, S.~N. (2011).
\newblock Fast stable restricted maximum likelihood and marginal likelihood estimation of semiparametric generalized linear models.
\newblock {\em Journal of the Royal Statistical Society: Series B (Statistical Methodology)} {\bf 73,} 3--36.

\bibitem[\protect\citeauthoryear{Wood}{Wood}{2017}]{wood2017generalized}
Wood, S.~N. (2017).
\newblock {\em Generalized Additive Models: An Introduction With R, Second Edition}.
\newblock CRC press.

\bibitem[\protect\citeauthoryear{Zhang, Beesoon, Zhu, and Martin}{Zhang et~al.}{2013}]{zhang2013biomonitoring}
Zhang, Y., Beesoon, S., Zhu, L., and Martin, J.~W. (2013).
\newblock Biomonitoring of perfluoroalkyl acids in human urine and estimates of biological half-life.
\newblock {\em Environmental Science \& Technology} {\bf 47,} 10619--10627.

\end{thebibliography}


\begin{thebibliography}{}

\bibitem[\protect\citeauthoryear{Andrews}{Andrews}{1994}]{andrews1994asymptotics}
Andrews, D. W.~K. (1994).
\newblock Asymptotics for semiparametric econometric models via stochastic equicontinuity.
\newblock {\em Econometrica} {\bf 62,} 43--72.

\bibitem[\protect\citeauthoryear{Belloni, Chernozhukov, Fernández-Val, and Hansen}{Belloni et~al.}{2017}]{belloni2018program}
Belloni, A., Chernozhukov, V., Fernández-Val, I., and Hansen, C. (2017).
\newblock Program evaluation and causal inference with high-dimensional data.
\newblock {\em Econometrica} {\bf 85,} 233--298.

\bibitem[\protect\citeauthoryear{Chernozhukov, Chetverikov, Demirer, Duflo, Hansen, Newey, and Robins}{Chernozhukov et~al.}{2018}]{dml}
Chernozhukov, V., Chetverikov, D., Demirer, M., Duflo, E., Hansen, C., Newey, W., and Robins, J. (2018).
\newblock {Double/debiased machine learning for treatment and structural parameters}.
\newblock {\em The Econometrics Journal} {\bf 21,} C1--C68.

\bibitem[\protect\citeauthoryear{Chipman, George, and McCulloch}{Chipman et~al.}{2010}]{chipman2010bart}
Chipman, H.~A., George, E.~I., and McCulloch, R.~E. (2010).
\newblock {BART: Bayesian additive regression trees}.
\newblock {\em The Annals of Applied Statistics} {\bf 4,} 266 -- 298.

\bibitem[\protect\citeauthoryear{Dorie}{Dorie}{2024}]{dbarts}
Dorie, V. (2024).
\newblock {\em dbarts: Discrete Bayesian Additive Regression Trees Sampler}.
\newblock R package version 0.9-26.

\bibitem[\protect\citeauthoryear{Dupont, Wood, and Augustin}{Dupont et~al.}{2022}]{dupont2022spatial+}
Dupont, E., Wood, S.~N., and Augustin, N.~H. (2022).
\newblock Spatial+: A novel approach to spatial confounding.
\newblock {\em Biometrics} {\bf 78,} 1279--1290.

\bibitem[\protect\citeauthoryear{Eberts and Steinwart}{Eberts and Steinwart}{2013}]{eberts2013optimal}
Eberts, M. and Steinwart, I. (2013).
\newblock {Optimal regression rates for SVMs using Gaussian kernels}.
\newblock {\em Electronic Journal of Statistics} {\bf 7,} 1 -- 42.

\bibitem[\protect\citeauthoryear{Gilbert, Datta, Casey, and Ogburn}{Gilbert et~al.}{2021}]{gilbert}
Gilbert, B., Datta, A., Casey, J.~A., and Ogburn, E.~L. (2021).
\newblock A causal inference framework for spatial confounding.
\newblock {\em arXiv preprint arXiv:2112.14946} .

\bibitem[\protect\citeauthoryear{Guan, Page, Reich, Ventrucci, and Yang}{Guan et~al.}{2022}]{guan_2020}
Guan, Y., Page, G.~L., Reich, B.~J., Ventrucci, M., and Yang, S. (2022).
\newblock {Spectral adjustment for spatial confounding}.
\newblock {\em Biometrika} {\bf 110,} 699--719.

\bibitem[\protect\citeauthoryear{Guinness}{Guinness}{2018}]{GpGp}
Guinness, J. (2018).
\newblock {Permutation and grouping methods for sharpening Gaussian Process approximations}.
\newblock {\em Technometrics} {\bf 60,} 415--429.
\newblock PMID: 31447491.

\bibitem[\protect\citeauthoryear{Huiying~Mao and Reich}{Huiying~Mao and Reich}{2023}]{model_free}
Huiying~Mao, R.~M. and Reich, B.~J. (2023).
\newblock Valid model-free spatial prediction.
\newblock {\em Journal of the American Statistical Association} {\bf 0,} 1--11.

\bibitem[\protect\citeauthoryear{Kanagawa, Hennig, Sejdinovic, and Sriperumbudur}{Kanagawa et~al.}{2018}]{kanagawa2018gaussian}
Kanagawa, M., Hennig, P., Sejdinovic, D., and Sriperumbudur, B.~K. (2018).
\newblock Gaussian processes and kernel methods: A review on connections and equivalences.
\newblock {\em arXiv preprint arXiv:1807.02582} .

\bibitem[\protect\citeauthoryear{Paciorek}{Paciorek}{2010}]{paciorek2010importance}
Paciorek, C.~J. (2010).
\newblock The importance of scale for spatial-confounding bias and precision of spatial regression estimators.
\newblock {\em Statistical Science} {\bf 25,} 107--125.

\bibitem[\protect\citeauthoryear{Rasmussen and Williams}{Rasmussen and Williams}{2005}]{williams2006gaussian}
Rasmussen, C.~E. and Williams, C. K.~I. (2005).
\newblock {\em {Gaussian Processes for Machine Learning}}.
\newblock The MIT Press.

\bibitem[\protect\citeauthoryear{Reich, Hodges, and Zadnik}{Reich et~al.}{2006}]{reich2006effects}
Reich, B.~J., Hodges, J.~S., and Zadnik, V. (2006).
\newblock Effects of residual smoothing on the posterior of the fixed effects in disease-mapping models.
\newblock {\em Biometrics} {\bf 62,} 1197--1206.

\bibitem[\protect\citeauthoryear{Robinson}{Robinson}{1988}]{robinson1988root}
Robinson, P.~M. (1988).
\newblock Root-n-consistent semiparametric regression.
\newblock {\em Econometrica} {\bf 56,} 931--954.

\bibitem[\protect\citeauthoryear{Stein}{Stein}{1999}]{stein1999interpolation}
Stein, M. (1999).
\newblock {\em Interpolation of Spatial Data: Some Theory for Kriging}.
\newblock Springer Series in Statistics. Springer New York.

\bibitem[\protect\citeauthoryear{Thaden and Kneib}{Thaden and Kneib}{2018}]{thaden2018structural}
Thaden, H. and Kneib, T. (2018).
\newblock Structural equation models for dealing with spatial confounding.
\newblock {\em The American Statistician} {\bf 72,} 239--252.

\bibitem[\protect\citeauthoryear{Wood}{Wood}{2011}]{wood2011gam}
Wood, S.~N. (2011).
\newblock Fast stable restricted maximum likelihood and marginal likelihood estimation of semiparametric generalized linear models.
\newblock {\em Journal of the Royal Statistical Society: Series B (Statistical Methodology)} {\bf 73,} 3--36.

\bibitem[\protect\citeauthoryear{Wood}{Wood}{2017}]{wood2017generalized}
Wood, S.~N. (2017).
\newblock {\em Generalized Additive Models: An Introduction With R, Second Edition}.
\newblock CRC press.

\end{thebibliography}









\label{lastpage}

\end{document}



\maketitle 
\renewcommand\thesection{S\arabic{section}}

\section{Overview of related semiparametric theory}
\subsection{Overview of related semiparametric theory}\label{s:literature}
\renewcommand\thealgorithm{S\arabic{algorithm}}

This section briefly summarizes the semiparametric literature on root-$n$ consistency and asymptotic normality of estimators similar to gSEM. More complete explanations can be found in \cite{andrews1994asymptotics} and \cite{dml}, from which this summary draws heavily. For $i\in\{1,...,n\}$, let $Y_i\in \mathbb{R}$ be the response variable, $A_i$ be the treatment variable, and $\bS_i$ be the spatial location contained in spatial domain $\mathcal{S}\subset \mathbb{R}^d$. A simple model is:
\begin{align}
    \begin{split}\label{e:intro_model}
        Y_i &= A_i\beta_0 + g_0(\bS_i) + U_i\\
        A_{i} &= m_0(\bS_i) + V_{i}.
    \end{split}
\end{align}
For illustration in this section we ignore covariates and assume $A_i$ is a scalar. In (\ref{e:intro_model}), $\beta_0$ is the regression coefficient of interest, $g_0$ and $m_0$ are unknown functions treated as nuisance parameters, and $U_i$ and $V_{i}$ are error terms with finite, non-zero variance such that $E(U_i|A_i, \bS_i) = 0$ and $E(V_i|\bS_i)=0$ for $i=1, ..., n$. 
Denote by $h_0(\bs) := E(Y_i|\bS_i=\bs)$ (marginalizing over $A_i$), and let $\eta_0$ denote the vector of functions $(g_0, m_0)$ or $(h_0, m_0)$. 

Model (\ref{e:intro_model}) can result in invalid inference due to spatial confounding if $g_0, m_0$ result in dependence between $\bA$ and $g_0(\bS)$. For example, if $g_0=m_0$, this corresponds to an unmeasured confounder which is a function of space affecting both $Y_1, ..., Y_n$ and $A_1, ..., A_n$, and in a spatial linear mixed model, this dependence is typically ignored. For example, spatial random effects are commonly used to implement a smoother to model the unknown function $g_0$, but are usually marginalized into the conditional variance of $\bY$ without consideration of dependence with $\bA$ \citep{paciorek2010importance}. Alternatively, some basis functions used to model $g_0$ may be collinear with $\bA$, in which case penalized estimation of the coefficients on the basis functions leads to bias in estimation of $\beta_0$ \citep{reich2006effects}.

\begin{algorithm}[H]\caption{Geoadditive Structural Equation Model (gSEM) estimation of $\beta_0$}\label{a:gsem}
\begin{algorithmic}
    \Require Response vector $\bY \in \mathbb{R}^{n}$, location matrix $\bS \in \mathbb{R}^{n \times d}$, treatment vector $\bA \in \mathbb{R}^{n}$
    \Ensure Estimate $\widehat{\beta}_{gSEM}$ of $\beta_0 \in \mathbb{R}$
    \State $\widehat{h}_0$ $\leftarrow$ Estimate of $h_0$ from a spatial regression model for $\bY$ onto $\bS$, such as a spline regression
    \State $\bR_Y \leftarrow \bY - \widehat{h}_0(\bS)$
    \State $\widehat{m}_0 \leftarrow$ Estimate of $m_0$ from a spatial regression model for $\bA$ onto $\bS$, such as a spline regression
    \State $\bR_A \leftarrow \bA - \widehat{m}_0(\bS)$
    \State $\widehat{\beta}_{gSEM} \leftarrow (\bR_A^T\bR_A)^{-1}\bR_A^T\bR_Y$, i.e., obtain $\widehat{\beta}_{0}$ by regressing the residuals $\bR_Y$ onto the residuals $\bR_A$
    \State Return $\widehat{\beta}_{gSEM}$
\end{algorithmic}
\end{algorithm}
The gSEM procedure is in Algorithm \ref{a:gsem}. Variance is typically estimated via bootstrap when gSEM is used in simulation studies \citep{guan_2020, gilbert}. Aside from the method of estimating $h_0$ and $m_0$ and variance estimation, this is identical to the procedure considered in \cite{robinson1988root}. Intuitively, it is approximately orthogonalizing $\bY$ and $\bA$ with respect to $\bS$ and therefore removing effects of spatial confounding. 

\subsubsection{Orthogonality of $\widehat{\beta}_{gSEM}$ and nuisance parameter estimates}\label{s:orthogonality}
To study $\widehat{\beta}_{gSEM}$'s asymptotic properties, note that $\widehat{\beta}_{gSEM}$ is equivalently defined as the solution to
\begin{align*}
    \frac{1}{n}\sum_{i=1}^n \psi(Y_i, A_i, \bS_i; \widehat{\beta}_{gSEM}, \widehat{\eta}_0) = 0
\end{align*}
where $\psi(Y_i, A_i, \bS_i; \beta, \eta)=\{Y_i - h(\bS_i) - \beta(A_i-m(\bS_i))\}(A_i-m(\bS_i))$, $\eta=(h,m)$, $\widehat{\eta}_0 = (\widehat{h}_0, \widehat{m}_0)$, and $\widehat{h}_0$ and $\widehat{m}_0$ are preliminary estimates of the nuisance parameters $h_0$ and $m_0$. The score function $\psi$ is used in \cite{robinson1988root}, and also studied in \cite{andrews1994asymptotics} and \cite{dml}. This score function results in a form of orthogonality between $\widehat{\beta}_{gSEM}$ and $\widehat{\eta}_0$: under the model (\ref{e:intro_model}), replacement of $\eta_0$ by $\widehat{\eta}_0$ in $\frac{1}{n}\sum_{i=1}^n E[\psi(Y_i, A_i, \bS_i; \beta_0, \eta_0)]$ has an effect that is $o_P(n^{-1/2})$ when $\widehat{\eta}_0$ is close to $\eta_0$ and $\widehat{h}_0$, $\widehat{m}_0$ each converge to their true values at $o_P(n^{-1/4})$ rate and obey smoothness conditions, and other regularity conditions are assumed \citep{andrews1994asymptotics}.\footnote{Estimating the nuisance parameters using spline regression may not meet the sufficient conditions for good asymptotic behavior presented in \citep{andrews1994asymptotics}, which used nonparametric kernel regression estimators.} Alternatively, \cite{dml} uses the term (near-) Neyman orthogonality. Neyman orthogonality means that the Gateaux functional derivative\footnote{Using the notation of \cite{dml}, the Gateaux derivative is defined as $D_r[\eta - \eta_0] := \partial_r\biggl\{E\bigl[\psi(W_; \beta_0,\eta_0 + r(\eta - \eta_0)\bigr]\biggr\}$ for $r \in [0, 1)$, $\eta \in T$ where $T$ is a convex subspace of a normed vector space, and $\partial_{\eta}E[\psi(W; \beta_0, \eta_0)]:=D_0[\eta-\eta_0]$.} with respect to $\eta$ is $0$ at the true nuisance parameter values:
\[ \partial_{\eta}E[\psi(Y_i, A_i, \bS_i; \beta_0, \eta_0)][\eta - \eta_0]=0, \]
and in the case of near-Neyman orthogonality, that it is $o_P(n^{-1/2})$. Similarly, this indicates that close to $\eta_0$, there is little effect on $E[\psi(Y_i, A_i, \bS_i; \beta_0, \eta_0)]$ when replacing $\eta_0$ by its estimate. The orthogonality property means that the estimates $\widehat{h}_0$ and $ \widehat{m}_0$ can converge to $h_0$ and $m_0$ at rates slower than $o_p(n^{-1/2})$ and asymptotically this deviation from $(h_0, m_0)$ does not affect affect the variance of $\widehat{\beta}_{gSEM}$ \citep{andrews1994asymptotics}.

\subsubsection{Stochastic equicontinuity or sample splitting}

The orthogonality property of $\psi$ must be paired with a means of ensuring that estimation error of $\eta_0$ does not cause asymptotic bias, primarily due to overfitting \citep{dml}. This is achieved by a property of empirical processes called stochastic equicontinuity in \cite{andrews1994asymptotics} and by sample splitting in \cite{dml}. Stochastic equicontinuity follows from Donsker conditions \citep{belloni2018program} which limit the complexity of $\eta_0$ and $\widehat{\eta}_0$. In \cite{andrews1994asymptotics}, these are primarily satisfied by placing a smoothness requirement on $\eta_0$ and $\widehat{\eta}_0$. In contrast, the sample-splitting approach of \cite{dml} avoids requiring Donsker conditions, and therefore additional smoothness requirements. Since by using sample splitting \cite{dml} allows estimation of $\eta_0$ by essentially any machine learning model, they term their method Double Machine Learning (DML). By using cross-fitting, the DML estimators of \cite{dml} still use the full sample so do not lose power. We therefore rely on \cite{dml} for our theoretical analysis of an estimator similar to gSEM, and term this estimator Double Spatial Regression since we are using Double Machine Learning with spatial regression. 


\section{Double Spatial Regression estimator for theoretical analysis}\label{s:DML_algorithm}

In this section, we consider an algorithm similar to gSEM {Appendix B} \citep{thaden2018structural} and the estimator in \cite{robinson1988root}, which allows derivation of explicit regularity conditions. 
For theoretical analysis, we combine treatment variables $\bA_i$ and covariates $\bZ_i$ into a combined vector $\bX_i$ of regressors of length $p$. The model considered is then, for $i=1,...,n$ and $j=1, ..., p$:
\begin{align}
\begin{split}\label{e:model_theory}
    Y_i &= \bX_i^T\bbeta_0 + g_0(\bS_i) + U_i,\quad \quad E(U_i|\bX_i, \bS_i)=0\\
    X_{ij} &= m_{0j}(\bS_i) + V_{ij}, \quad \quad \quad \quad \quad E(V_{ij}|\bS_i=0).
\end{split}
\end{align}

{We treat covariates as treatment variables since existing results from machine learning theory, such as \cite{eberts2013optimal}, do not typically incorporate parametric adjustment for covariates into their analysis. We use these existing results to fully justify and analyze our estimator. We are not aware of available theoretical results that would allow parametric adjustment for covariates within our framework.}
DSR uses Kriging (although other nonparametric estimators can be used) which requires a working correlation function be specified, and for theoretical analysis we use the {squared exponential} correlation function \citep{williams2006gaussian}:
\begin{align*}
    C_{\gamma}(\bS_i, \bS_j):= \exp\biggl(-\frac{\|\bS_i - \bS_j\|_2^2}{\gamma^2}\biggr).
\end{align*}
For matrix inputs $\bA$ and $\bB$ to a correlation function, $C(\bA, \bB)$ denotes the correlation matrix where the element in row $i$ and column $j$ equals $C(\bA_{i\cdot}, \bB_{j \cdot})$.

Denote by $h_0(\bs)$ the conditional expectation $E(Y_i|\bS_i=\bs)$ (note that this does not depend on $\bX_i$). The number of elements in fold $k$ is denoted by $|\bk|$. On each fold $k$, Kriging \citep{stein1999interpolation} is used to obtain cross-fitted estimates for $h_0(\bS_{\bk})$ and $m_{0j}(\bS_{\bk})$ for $j=1,...,p$:
\begin{align}
    \widehat{h}_0(\bS_{\bk}) := C_{\gamma_{0k}}(\bS_{\bk}, \bS_{\bk^C})\biggl(C_{\gamma_{0k}}(\bS_{\bk^C}, \bS_{\bk^C}) +|\bk^C|\lambda_{0k} \bI\biggr)^{-1}\bY_{\bk^C}\label{eq:kriging_eqn1}\\
    \widehat{m}_{0j}(\bS_{\bk}) := C_{\gamma_{jk}}(\bS_{\bk}, \bS_{\bk^C})\biggl(C_{\gamma_{jk}}(\bS_{\bk^C}, \bS_{\bk^C}) + |\bk^C|\lambda_{jk} \bI\biggr)^{-1}\bX_{\bk^C, j} \label{eq:kriging_eqn2}.
\end{align}
Predictions are combined across folds to obtain $\widehat{h}_0(\bS)$ and $\widehat{m}_{01}(\bS), ..., \widehat{m}_{0p}(\bS)$. 
The hyperparameters $\lambda_{0k}, ..., \lambda_{pk}$, $\gamma_{0k}, ..., \gamma_{pk}$ depend on $k$ because our theoretical analysis requires that they are selected each time predictions are obtained (i.e., on each fold); they are selected using a training-validation split of $\bW_{\bk^C}$ \citep{eberts2013optimal}. The data in $\bW_{\bk^C}$ is split in two halves, a ``training" half and a ``validation" half. For each possible combination of hyperparameters ($\lambda$ and $\gamma$) considered, the training half of $\bW_{\bk^C}$ is used to obtain predictions on the validation half of $\bW_{\bk^C}$. The pair of hyperparameters with lowest mean squared error (MSE) on the validation half of $\bW_{\bk^C}$ is selected for estimating $h_0, m_{01}, ..., m_{0p}$ evaluated at the locations corresponding to the data in $\bW_{\bk}$ \citep{eberts2013optimal}. {This procedure is adopted in order to satisfy the requirements of \cite{eberts2013optimal}, which proves the convergence rate of certain GP regression estimates, in order to allow derivation of clear regularity conditions for DSR. }

Letting $\widehat{\bV}_{\cdot j} = \bX_{\cdot j} - \widehat{m}_{0j}(\bS)$ and $\widehat{U}_i = Y_i - \widehat{h}_0(\bS_i) - \widehat{\bV}_i^T\widehat{\bbeta}_0$, the DSR estimator and its approximate variance are:
\begin{align}
    \widehat{\bbeta}_0 &= (\widehat{\bV}^T\widehat{\bV})^{-1}\widehat{\bV}^T(\bY - \widehat{h}_0(\bS))\label{e:dsr_theory_estimator}\\
    \widehat{Var}(\widehat{\bbeta}_0) &= (\widehat{\bV}^T\widehat{\bV})^{-1} \sum_{i=1}^n \left[\widehat{U}_i^2\widehat{\bV}_i\widehat{\bV}_i^T\right](\widehat{\bV}^T\widehat{\bV})^{-1},\label{e:dsr_theory_var_estimator}
\end{align}
which are the estimators from \cite{dml}. The algorithm is presented in Algorithm \ref{a:SDML}. Note that this estimator is essentially the same as gSEM except that it uses sample splitting, the method of estimating the latent functions of space is not specified with gSEM, and gSEM has lacked a closed-form variance estimate; gSEM in turn is essentially identical to the estimator in \cite{robinson1988root}, except that \cite{robinson1988root} uses nonparametric kernel regression estimators and provides a variance estimate.

\begin{algorithm}\caption{Double Spatial Regression estimation of $\beta_0$ for theoretical study}\label{a:SDML}
\begin{algorithmic}
    \Require Centered response vector $\bY \in \mathbb{R}^{n}$, location matrix $\bS \in \mathbb{R}^{n \times d}$, design matrix $\bX \in \mathbb{R}^{n \times p}$ with centered columns.
    \Ensure Estimate $\widehat{\bbeta}_0$ of $\bbeta_0 \in \mathbb{R}^p$ and estimate $\widehat{Var}(\widehat{\bbeta}_0)$ of $Var(\widehat{\bbeta}_0) \in \mathbb{R}^{p\times p}$
    \State Randomly partition the data into $K$ folds so that the size of each fold is $\frac{n}{K}$. 
    \For{$k=1, ..., K$}
    \State Select the hyperparameters $\gamma_{0,k}, \lambda_{0,k} \in \mathbb{R}$ by a training-validation approach, using data in $\bk^C$. The value of $\lambda_{0k}$ is selected from an evenly-spaced grid of $\frac{1}{2n}$ values in $(0, 1]$, the lengthscale $\gamma_{0k}$ is selected from an evenly-spaced grid of $\frac{1}{2n^{-1/4}}$ values in $(0,1]$.
    \State $\widehat{h}_0(\bS_{\bk}) \leftarrow C_{\gamma_{0k}}(\bS_{\bk}, \bS_{\bk^C})\biggl(C_{\gamma_{0k}}(\bS_{\bk^C}, \bS_{\bk^C}) +|\bk^C|\lambda_{0k} \bI\biggr)^{-1}\bY_{\bk^C}$
    \For{$j=1, ..., p$}
    \State Select the hyperparameters $\gamma_{j,k}, \lambda_{j,k} \in \mathbb{R}$ by a training-validation approach, using data in $\bk^C$. The value of $\lambda_{jk}$ is selected from an evenly-spaced grid of $\frac{1}{2n}$ values in $(0, 1]$, the lengthscale $\gamma_{jk}$ is selected from an evenly-spaced grid of $\frac{1}{2n^{-1/4}}$ values in $(0,1]$.
    \State $\widehat{m}_{0j}(\bS_{\bk}) \leftarrow C_{\gamma_{jk}}(\bS_{\bk}, \bS_{\bk^C})\biggl(C_{\gamma_{jk}}(\bS_{\bk^C}, \bS_{\bk^C}) + |\bk^C|\lambda_{jk} \bI\biggr)^{-1}\bX_{\bk^C, j}$
    \EndFor
    \EndFor
    
    \State Combine the predictions from each fold to obtain $\widehat{h}_0(\bS) \in \mathbb{R}^{n}, \widehat{\bmm}_0(\bS) \in \mathbb{R}^{n\times p}$ 
    \State $\widehat{\bV} \leftarrow \bX - \widehat{\bmm}_0(\bS)$ 
    \State $\widehat{\mathbf{J}} \leftarrow (\widehat{\bV}^T\widehat{\bV})^{-1}$
    \State $\widehat{\bbeta}_0 \leftarrow \widehat{\mathbf{J}}\widehat{\bV}^T(\bY - \widehat{h}_0(\bS))$
    
    \State $\widehat{Var}(\widehat{\bbeta}_0) \leftarrow \widehat{\mathbf{J}} \left(\sum_{i=1}^n \left[-\widehat{\bV}_{i}\widehat{\bV}_{i}^T\widehat{\bbeta}_0 + \widehat{\bV}_i (Y_i - \widehat{h}_0(\bS_i))\right]^2\right)\widehat{\mathbf{J}}^{T} $
    \State Return $\widehat{\bbeta}_0, \widehat{Var}(\widehat{\bbeta}_0)$
\end{algorithmic}
\end{algorithm}

\subsection{Double Spatial Regression regularity conditions}\label{s:regularity}
 
Our theoretical analysis uses the results on convergence rates of GP regression from \cite{eberts2013optimal} and asymptotic properties of DML estimators from \cite{dml}, which require fast-enough convergence of estimates of the latent functions, to obtain explicit regularity conditions on the latent functions of space $h_0$ and $m_{01}, ..., m_{0p}$ under which DSR is root-$n$ asymptotically normal and consistent. The observations $\bW_i=(Y_i, \bX_i, \bS_i)$ are assumed to be i.i.d. from a probability distribution $P$ with density function $p(w)$, and for a function $f$, $\| f \|_{P,q} := \{ \int |f(w)|^qp(w)dw\}^{1/q}$. {Note that although the $\bW_i$ are assumed to be drawn i.i.d. from $P$, the distribution of $Y_i|\bX_i$ is still assumed to, in general, exhibit spatial dependence induced by marginalization over $\bS_i$.} The Euclidean norm is denoted $\|\cdot\|$, and $\|\cdot\|_p$ and $L_p(\mathbb{R}^d)$ are with respect to the Lebesgue measure.


As the assumptions govern the probability distribution $P$ generating the i.i.d. random variables $\bW_i$, regularity conditions are described for all draws, indexed by $i$, from $P$. The assumptions for the DSR estimator obtained by Algorithm \ref{a:SDML} are:

\begin{enumerate}
    \item[A1)] The data are generated by (\ref{e:model_theory}).
    \item[A2)] The errors $U_i, \bV_i$ are such that $E(U_i|\bX_i, \bS_i) = E(V_{ij}|\bS_i) = 0$, for $j=1,...,p$, with $0 < E(U_i^2|\bS_i) \leq C$ and $0 < E(V_{ij}^2|\bS) \leq C$ for some constant $C>0$ and all $\bS_i \in \mathcal{S}$. Also, $E(U_i^2\bV_i\bV_i^T)$ and $E(\bV_i\bV_i^T)$ have minimum eigenvalues bounded away from 0. The errors $U_i$ and $V_{ij}$ are either contained in some interval or are normally distributed.
    \item[A3)] The spatial locations $\bS_i$ reside in a region $\mathcal{S}$ contained in a $\|\cdot \|$-unit ball in $\mathbb{R}^d$, and the boundary of $\mathcal{S}$ has $P$-probability 0. The marginal distribution $P_S$ of $\bS_i$ (derived from $P$) is absolutely continuous on $\calS$ and has a density $p_S \in L_q(\mathbb{R}^d)$ for some $q \geq 1$. 
    \item[A4)] If $U_i$ is bounded, then $h_0$ is such that $Y \in [-M_0, M_0]$ for some $M_0>0$, and if $U_i$ is normally distributed then $h_0 \in [-1, 1]$. Similarly, if $V_{ij}$ is bounded, then $m_{0j}$ is such that $X_{ij} \in [-M_j, M_j]$ for some $M_j > 0$, and if $V_{ij}$ is normally distributed then $m_{0j} \in [-1, 1]$, for $j=1, ..., p$.
    \item[A5)] The estimates $\widehat{h}_0, \widehat{m}_{01}, ..., \widehat{m}_{0p}$ are obtained using GP regression as described in \cite{eberts2013optimal}. The function $\widehat{h}_0$ is clipped so that if $U_i$ is bounded, $|\widehat{h}_0| \leq M_0$ and if $U_i$ is normally distributed, $|\widehat{h}_0| \leq \min \{1, 4\sqrt{C_0}\sqrt{\ln(n)} \}$ for some $C_0 > 0$ that exceeds $Var(U_i | \bS_i)\, \forall \bS_i \in \mathcal{S}$. Similarly, for $j=1,...,p$, $\widehat{m}_{0j}$ is clipped so that if $V_{ij}$ is bounded, $|\widehat{m}_{0j}| \leq M_j$ and if $V_{ij}$ is normally distributed, $|\widehat{m}_{0j}| \leq \min \{1, 4\sqrt{C_j}\sqrt{\ln(n)} \}$ for some $C_j > 0$ that exceeds $Var(V_{ij} | \bS)\, \forall \bS_i \in \mathcal{S}$.
    \item[A6)] For $j=1,...,p$, the functions $m_{0j}$ reside in the Besov space $B^{\alpha_X}_{2s, \infty}$ where the smoothness order $\alpha_X > \frac{d}{2}$, $\alpha_X \geq 1$, and $\frac{1}{s} + \frac{1}{q} = 1$ and $s \geq 1$, and $h_0$ resides in the Besov space $B^{\alpha_Y}_{2s, \infty}$ where the smoothness order $\alpha_Y > \frac{d^2}{4\alpha_X}$ and $\alpha_Y \geq 1$. Furthermore, $h_0, m_{01}, ..., m_{p0} \in L_2(\mathbb{R}^d)\cap L_{\infty}(\mathbb{R}^d)$.
\end{enumerate}

In the common scenario $d=2$, a stronger but more interpretable alternative to Assumption A6 is that $h_0, m_{01}, ..., m_{0p}$ are each in $L_2(\mathbb{R}^d)\cap L_{\infty}(\mathbb{R}^d)$, each has at least two (weak) derivatives, and these functions and derivatives are all in $L_{2s}(\mathbb{R}^d)$. Per the discussion following Theorem 3.6 in \cite{eberts2013optimal}, $U_i, V_{i1}, ..., V_{ip}$ can follow other light-tailed distributions aside from normal.


Theorem \ref{t:1} states that the DSR estimator in (\ref{e:dsr_theory_estimator}) is root-$n$ asymptotically normal and consistent under the above assumptions.
\begin{theorem}\label{t:1}
If Assumptions A1 -- A6 are met, and $\widehat{\bbeta}_0$ and $\widehat{Var}(\widehat{\bbeta}_0)$ are obtained by Algorithm \ref{a:SDML}, then
\begin{align*}
 \sqrt{n}\boldsymbol{\Sigma}^{-1/2}(\widehat{\bbeta}_0- \bbeta_0) \xrightarrow[]{D} N(\mathbf{0}, \mathbf{I}_p), \text{and} \\
 \widehat{Var}(\widehat{\bbeta}_0)^{-1/2}(\widehat{\bbeta}_0 - \bbeta_0) \xrightarrow[]{D} N(\mathbf{0},\mathbf{I}_p),  
\end{align*}
where $\boldsymbol{\Sigma}=E[ \bV_i\bV_i^T]^{-1}E[U_i^2\bV_i\bV_i^T](E[ \bV_i\bV_i^T]^{-1})$ is the approximate variance of $\widehat{\bbeta}_0$.

\end{theorem}
The proof of Theorem \ref{t:1} is in the Supplementary Materials Section S6.

\subsection{Smoothness conditions}

To provide intuition on Assumption A6, the following explanation of Besov spaces is paraphrased from \cite{eberts2013optimal}. Denote the $\zeta$-th weak derivative $\partial^{(\zeta)}$ for a multi-index $\zeta = (\zeta_1,\zeta_2, ..., \zeta_d) \in \mathbb{N}^d$ with $|\zeta|=\sum_{i=1}^d\zeta_i$. With regard to a measure $\nu$, the Sobolev space $W_p^\alpha(\nu)$ is defined as: 
\begin{align*}
W_p^{\alpha}(\nu) :=  \{f \in L_p(\nu):\partial^{(\zeta)}f\in L_p(\nu) \text{ exists for all $\zeta \in \mathbb{N}^d$ with $|\zeta|<\alpha$}\}.    
\end{align*}
Loosely speaking, $W_p^{\alpha}(\nu)$ is the space of functions with $\alpha$ weak derivatives, which all must have finite $L_p(\nu)$ norm.
We refer to \cite{eberts2013optimal} for a full definition of Besov spaces $B_{p,q}^\alpha$, but Besov spaces provide a finer scale of smoothness than the integer-ordered Sobolev spaces, and Sobolev spaces are contained in the Besov spaces: \[W_p^\alpha(\mathbb{R}^d) \subset B_{p,q}^\alpha(\mathbb{R}^d)  \] for $\alpha \in \mathbb{N}, p\in(1, \infty), \max\{p, 2\} \leq q \leq \infty$. 

A stronger, but more interpretable alternative to the assumption that $h_0 \in B^{\alpha_Y}_{2s, \infty}$ and $m_{01}, ..., m_{0p} \in B^{\alpha_X}_{2s, \infty}$ in assumption A6 is that the $m_{0j}$ reside in the integer-order Sobolev space $W_{2s}^{\alpha_X'}$ and $h_0$ resides in the integer-order Sobolev space $W_{2s}^{\alpha_Y'}$, where $\alpha_X'=\lceil \alpha_X \rceil$ and $\alpha_Y' = \lceil \alpha_Y \rceil$ and $\lceil a \rceil$ indicates the lowest integer greater than $a$. In the common scenario $d=2$, the requirements on $\alpha_X$ and $\alpha_Y$ reduce to $\alpha_X > 1$ and $\alpha_Y>1$, which for integer-ordered Sobolev spaces, loosely means that both $h_0$ and $m_0$ have at least two partial derivatives, or are smoother than functions with only one partial derivative. 

\section{Additional algorithm}

Below is the algorithm for DSR estimation without cross-fitting. As in the main paper, the assumed model for $j = 1, ..., \ell$ is:
\begin{align}\label{eq:SDML_practical_model1}
    Y_i &= \bA_i^T\bbeta_0 + \bZ_i^T\boldsymbol{\theta}_{00} + g_0(\bS_i) + U_i\\
    A_{ij} &= \bZ_i^T\boldsymbol{\theta}_{0j} + m_{0j}(\bS_i) + V_{ij}, \label{eq:SDML_practical_model2}
\end{align}
If cross-fitting is not used, the Universal Kriging equations are:
\begin{align}\label{eq:univ_kriging_nosplit_eqn1}
    \widehat{g}_0(\bS) &= \widehat{\omega}^2_{0}C_{\widehat{\gamma}_{0},\widehat{\tau}_{0}}(\bS, \bS)\biggl(\widehat{\omega}^2_{0}C_{\widehat{\gamma}_{0},\widehat{\tau}_{0}}(\bS, \bS) + \widehat{\sigma}_{0}^2 \bI\biggr)^{-1}(\bY - \bA^T\widetilde{\bbeta}_{0} - \bZ^T \widehat{\boldsymbol{\theta}}_{0}) \\
     \widehat{m}_{0j}(\bS) &=\widehat{\omega}^2_{j}C_{\widehat{\gamma}_{j},\widehat{\tau}_{j}}(\bS, \bS)\biggl(\widehat{\omega}^2_{j}C_{\widehat{\gamma}_{j},\widehat{\tau}_{j}}(\bS, \bS) + \widehat{\sigma}_{j}^2 \bI\biggr)^{-1}(\bA_{ j} -\bZ^T \widehat{\boldsymbol{\theta}}_{j}),\label{eq:univ_kriging_nosplit_eqn2}
\end{align}
where notation is as in the main paper, except that the subscript $k$ indicating folds has been dropped as all parameter selections and predictions are obtained using the full dataset.

\begin{algorithm}[H]\caption{DSR estimation of $\bbeta_0$ without cross-fitting}\label{a:SDML_practical_nocross}
\begin{algorithmic}
    \Require Response vector $\bY \in \mathbb{R}^{n}$, location matrix $\bS \in \mathbb{R}^{n \times 2}$, treatment matrix $\bA \in \mathbb{R}^{n \times \ell}$, covariate matrix $\bZ \in \mathbb{R}^{n \times m}$.
    \Ensure Estimate $\widehat{\bbeta}_{DSR}$ of $\bbeta_0 \in \mathbb{R}^{\ell}$ and estimate $\widehat{Var}(\widehat{\bbeta}_{DSR})$ of $Var(\widehat{\bbeta}_{DSR}) \in \mathbb{R}^{\ell \times \ell}$
    \State Using $\bW$, obtain $\widetilde{\boldsymbol{\beta}}_{0},$ and for $j=0, ..., \ell$, $\widetilde{\boldsymbol{\theta}}_{0j}$,  $\widetilde{\gamma}_{j}, \widetilde{\tau}_{j}, \widetilde{\sigma}_{j}, \widetilde{\omega}_{j}$ by fitting Models (\ref{eq:SDML_practical_model1}) and (\ref{eq:SDML_practical_model2}) using \texttt{GpGp}.
    \State Obtain $\widehat{g}_0(\bS)$ by (\ref{eq:univ_kriging_nosplit_eqn1}) (approximated by \texttt{GpGp}).

    \For{$j=1, ..., \ell$}
    \State Obtain $\widehat{m}_{0j}(\bS)$ by (\ref{eq:univ_kriging_nosplit_eqn2}) (approximated by \texttt{GpGp}).
    \State $\widehat{\bA}_{\cdot, j} \leftarrow \bZ^T\widetilde{\boldsymbol{\theta}}_{0j} + \widehat{m}_{0j}(\bS) $

    \EndFor
    
    \State $\widehat{\bV} \leftarrow \bA - \widehat{\bA}$
    \State $\widehat{\bbeta}_{DSR} \leftarrow (\widehat{\bV}^T\bA)^{-1}\widehat{\bV}^T(\bY - \bZ^T {\widetilde{\boldsymbol{\theta}}}_{00} - \widehat{g}_0(\bS))$
    \State $\widehat{\bU} = \bY - \bA^T\widehat{\bbeta}_{DSR} - \bZ^T\widetilde{\boldsymbol{\theta}}_{00} - \widehat{g}_0(\bS)$
     \State $\widehat{Var}(\widehat{\bbeta}_{DSR}) \leftarrow (\widehat{\bV}^T\bA)^{-1} \left(\sum_{i=1}^n \widehat{U}_i^2\widehat{\bV}_i \widehat{\bV}_i^T\right)\left((\widehat{\bV}^T\bA)^{-1}\right)^T$
    \State Return $\widehat{\bbeta}_{DSR}, \widehat{Var}(\widehat{\bbeta}_{DSR})$
\end{algorithmic}
\end{algorithm}

\section{Additional simulation details}
Additional scenarios:
\begin{enumerate}
    \item Cubed confounder: $Y_i \sim N(\beta A_i + Z_i^3, \sigma^2_Y)$.
    \item Gamma errors in $\bY$: $Y_i = \beta A_i + Z_i + \phi_i, \phi_i =q[\Phi(\epsilon_i/\sqrt{3})] $, where $q$ is the quantile function for the Gamma$(1, 1/\sqrt{3})$ distribution and $\Phi$ is the standard normal CDF, and $\epsilon_i \sim N(0,\sigma_Y^2)$.
    \item ``East-west" heteroskedasticity: $Y_i = \beta A_i + Z_i + S_{1i}\epsilon_i$, where $S_{1i}$ is the first coordinate of $\bS_i$ and $\epsilon_i \sim N(0,\sigma_Y^2)$.
    \item ``Middle-out" heteroskedasticity: $Y_i = \beta A_i + \sqrt{\frac{\omega(S_{1i})}{3}}Z_i + \sqrt{1 - \omega(S_{1i})}\epsilon_i$, where $\omega(S_{1i}) = \Phi(\frac{S_{1i}-0.5}{0.1})$ and $\epsilon_i \sim N(0,\sigma_Y^2)$.
\end{enumerate}
The last three are borrowed from \cite{model_free}. Five further scenarios were considered.
\begin{enumerate}
    \item Higher variance in $\bA$: $\sigma^2_A=1$, causing less confounding bias. 
    \item Very rough processes:  $\boldsymbol{\Sigma}_A$ and $\boldsymbol{\Sigma}_Z$ were Mat\`ern correlation matrices with smoothness 0.5, equivalent to exponential covariance, making adjustment more challenging due to very rough sample paths. The range parameter for the covariance function was 0.114 to have similar practical spatial range as the other scenarios. 
    \item Gridded spatial locations: Theory requires random spatial locations, but this illustrates the method with (very regular) fixed spatial locations.
    \item Deterministic function of space, same for $\bA$ and $\bZ$: To avoid over-stating the effectiveness of DSR when the latent functions of space are generated and estimated using GPs, the data were generated using: $Y_i = \beta_0A_i + g_0(\bs_i) + \epsilon_{0i}$, and $A_i=m_0(\bs_i) + \epsilon_{1i}$, where $g_0(\bs_i) = m_0(\bs_i) = \cos(10s_{i1})\sin(10s_{i2})$, $\epsilon_{0i} \stackrel{i.i.d.}{\sim}N(0, 1^2)$, and $\epsilon_{1i} \stackrel{i.i.d.}{\sim}N(0, 0.1^2)$.
    \item Deterministic function of space, different for $\bA$ and $\bZ$: Similar to the previous scenario, but now $g_0(\bs_i) = m_0(\bs_i) + \sin(10s_{i1})\sin(10s_{i2})$ and $m_0$ is defined as in the previous scenario.
\end{enumerate}

Plots of smooth, rough, and very rough simulated spatial surfaces are below in Figures \ref{fig:smooth_surface}, \ref{fig:rough_surface}, and \ref{fig:very_rough_surface}. These are examples of surfaces drawn from the distributions used to generate observations of the treatment variable $\bA$ and the unobserved confounder $\bZ$ in the simulation study. Since these surfaces are drawn randomly in each iteration of the simulation, these are only representative of the level of smoothness in each type of distribution used, and are not actual simulated datasets used. The data in the following plots are placed on a regular grid for easier visualization of the smoothnesses, rather than having spatial locations drawn randomly as in most of the simulated scenarios.

\begin{figure}
    \centering
    \includegraphics[width=0.5\linewidth]{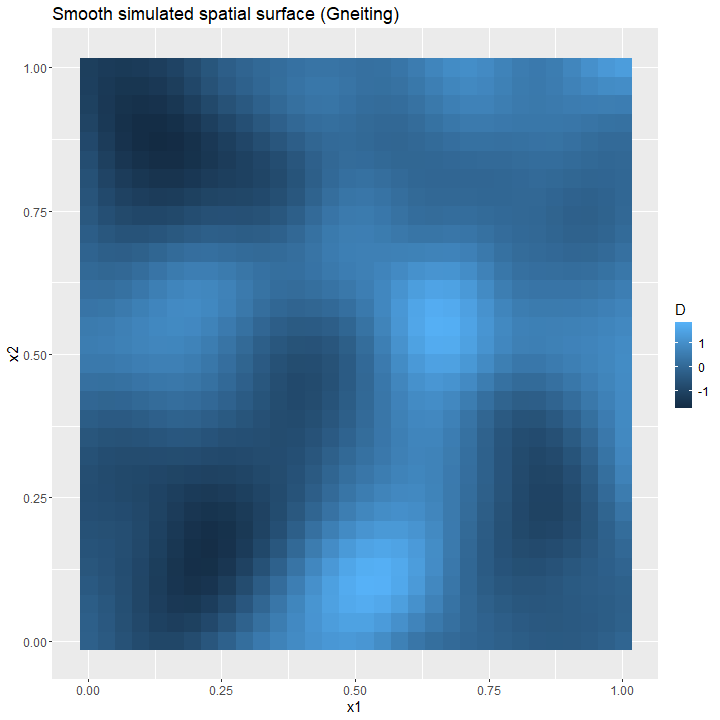}
    \caption{Smooth simulated spatial surface, drawn a multivariate normal distribution with Gneiting covariance, and range parameter 0.2, as used in the simulation study.}
    \label{fig:smooth_surface}
\end{figure}

\begin{figure}
    \centering
    \includegraphics[width=0.5\linewidth]{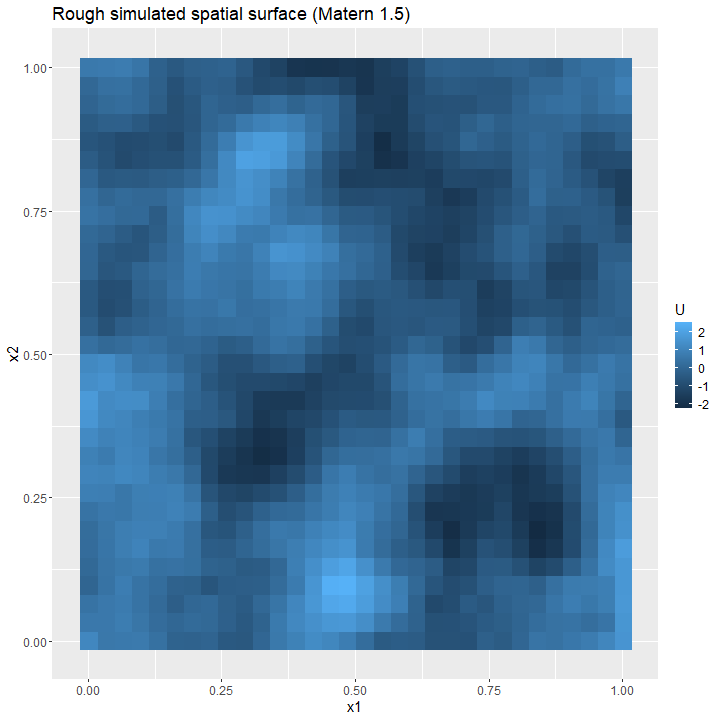}
    \caption{Rough simulated spatial surface, drawn a multivariate normal distribution with Mat\`ern covariance, smoothness parameter 1.5, and range parameter 0.072, as used in the simulation study.}
    \label{fig:rough_surface}
\end{figure}

\begin{figure}
    \centering
    \includegraphics[width=0.5\linewidth]{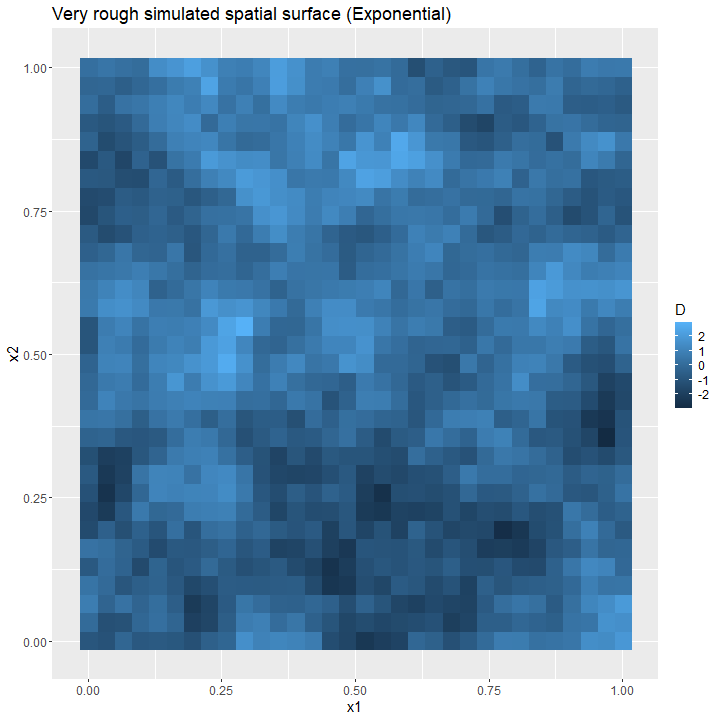}
    \caption{Very rough simulated spatial surface, drawn from a multivariate normal distribution with exponential covariance, and range parameter 0.114, as used in the simulation study.}
    \label{fig:very_rough_surface}
\end{figure}

In implementation of gSEM and Spatial+, no adjustment was made in the bootstrapping procedure for spatial correlation between observations, or the tendency of GAMs to under-smooth in bootstrap samples described in \cite{wood2017generalized}. Spline models, gSEM, and Spatial+ used 300 spline basis functions.

\section{Full simulation results}
\setcounter{table}{0}
\renewcommand{\thetable}{S\arabic{table}}
In the following tables, the following methods were compared:
\begin{itemize}
    \item OLS: ordinary least squares regression.
    \item LMM: Spatial linear mixed model, estimated using \texttt{GpGp} \citep{GpGp}.
    \item Spline (GCV): spline model estimated using \texttt{mgcv} \citep{wood2011gam}, with smoothing parameter selected by generalized cross-validation, to minimize out-of-sample prediction error.
    \item Spline (REML): spline model estimated using \texttt{mgcv}, smoothing parameter selected by restricted maximum likelihood (REML); \cite{wood2017generalized} states that this might reduce bias when the parametric component is collinear with smooth term.
    \item Spatial+: the method of \cite{dupont2022spatial+}.
    \item gSEM: the point-referenced version of \cite{thaden2018structural}, described in Algorithm \ref{a:gsem}.
    \item Shift (BART): the shift estimand implemented in \cite{gilbert}, using Bayesian Additive Regression Trees \citep{chipman2010bart}, or BART, to obtain the preliminary nonparametric estimates. BART estimates were obtained from the R package \texttt{dbarts} \cite{dbarts}. Results only presented in Table \ref{tab:sim_highvar} due to inability to obtain estimates in other scenarios. Variance estimates not obtained due to relatively high expense of bootstrapping.
    \item DSR (theory): the theoretical DSR estimator described in Algorithm \ref{a:SDML}.  
    \item DSR (theory, no crossfit): the theoretical DSR estimator described in Algorithm \ref{a:SDML} but without crossfitting.
    \item DSR: The DSR estimator described in Algorithm 1.
    \item DSR: (no crossfit): the DSR estimator without crossfitting described in Algorithm \ref{a:SDML_practical_nocross}.
    \item DSR (spline, no crossfit): the DSR estimator without crossfitting described in Algorithm \ref{a:SDML_practical_nocross} but with splines used instead of GP regression.
    \item DSR (theory, spline, no crossfit): the theoretical DSR estimator described in Algorithm \ref{a:SDML} but without crossfitting, and with splines used instead of GP regression. This is equivalent to the implementation of gSEM used, except that it has a closed-form variance estimate.
    \item DSR (theory, GpGp): the theoretical DSR estimator described in Algorithm \ref{a:SDML} but with GP regression implemented by the \texttt{GpGp} package \citep{GpGp}.
    \item DSR (Smooth covariance): The DSR estimator with crossfitting described in Algorithm 1, but with the smoothness term of the Mat\`ern covariance fixed at 4.5 for a smoother fitted function of space.
    \item DSR (Theory, smooth covariance): the theoretical DSR estimator described in Algorithm \ref{a:SDML} but with GP regression implemented by the \texttt{GpGp} package \citep{GpGp}, and with the smoothness term of the Mat\`ern covariance fixed at 4.5 for a smoother fitted function of space.
\end{itemize}
Note that for theoretical DSR, we use a grid of possible $\gamma$ values of size $\frac{1}{2n^{-1/2}}$ to improve finite-sample performance.
The following metrics were used to compare the methods:
\begin{itemize}
    \item Bias: average difference between estimates and $\beta_0$ over the Monte Carlo iterations. In all simulations $\beta_0=0.5$.
    \item Rel. Bias: relative bias, equal to bias divided by $\beta_0$. 
    \item MSE: mean squared error of estimates over the Monte Carlo iterations.
    \item 95\% CI Length: mean length of 95\% confidence intervals.
    \item 95\% CI CVG: coverage, i.e. proportion of Monte Carlo iterations in which the 95\% confidence interval included $\beta_0$.
    \item Power: Proportion of Monte Carlo iterations in which the 95\% confidence interval did not include 0.
\end{itemize}

\renewcommand{\arraystretch}{0.7}
\begin{table}[H]
    \centering
    \small
        \caption{Simulation results for rough $\bU$ and rough $\bA$. Metrics are bias, relative bias (bias divided by $\beta_0$), mean squared error (MSE), confidence interval (CI) length, coverage (CVG), and power. CI length, coverage, and power are computed with respect to 95\% confidence intervals.}
    \begin{tabular}[t]{lrrrrrrr}
\hline
  & Bias & Rel. Bias & MSE & CI Length & CVG & Power\\
OLS & 0.499 & 0.997 & 0.273 & 0.170 & 0.005 & 1.000\\
LMM & 0.474 & 0.947 & 0.232 & 0.329 & 0.000 & 1.000\\
Spline (GCV) & 0.469 & 0.938 & 0.228 & 0.358 & 0.000 & 1.000\\
Spline (REML) & 0.474 & 0.948 & 0.233 & 0.332 & 0.000 & 1.000\\
Spatial+ & 0.483 & 0.966 & 0.303 & 1.256 & 0.700 & 0.897\\
gSEM & 0.472 & 0.944 & 0.291 & 1.248 & 0.703 & 0.905\\
DSR (theory) & 0.452 & 0.905 & 0.215 & 0.425 & 0.005 & 1.000\\
DSR (theory, GpGp) & 0.301 & 0.602 & 0.141 & 0.786 & 0.625 & 0.965\\
DSR (theory, no crossfit) & 0.482 & 0.965 & 0.251 & 0.466 & 0.030 & 1.000\\
DSR & 0.321 & 0.642 & 0.152 & 0.791 & 0.623 & 0.940\\
DSR (no crossfit) & 0.291 & 0.582 & 0.132 & 0.812 & 0.655 & 0.940\\
DSR (spline, no crossfit) & 0.332 & 0.663 & 0.147 & 0.708 & 0.527 & 0.985\\
DSR (theory, spline, no crossfit) & 0.472 & 0.944 & 0.291 & 0.942 & 0.478 & 0.958\\
DSR (theory, GpGp, smooth) & \textbf{0.284} & 0.568 & \textbf{0.128} & 0.773 & 0.650 & 0.965\\
DSR (smooth) & 0.320 & 0.639 & 0.148 & 0.801 & 0.640 & 0.953\\

\hline

\end{tabular}

    \label{tab:sim_roughU_roughA}
\end{table}

\normalsize

\begin{table}[H]
    \centering
    \small
    \caption{Simulation results for rough $\bU$ and smooth $\bA$. Metrics are bias, relative bias (bias divided by $\beta_0$), mean squared error (MSE), confidence interval (CI) length, coverage (CVG), and power. CI length, coverage, and power are computed with respect to 95\% confidence intervals.}
    \begin{tabular}[t]{lrrrrrrr}
\hline
  & Bias & Rel. Bias & MSE & CI Length & CVG & Power\\
OLS & 0.432 & 0.863 & 0.222 & 0.179 & 0.038 & 1.000\\
LMM & 0.441 & 0.881 & 0.213 & 0.505 & 0.090 & 1.000\\
Spline (GCV) & 0.463 & 0.925 & 0.239 & 0.700 & 0.215 & 0.998\\
Spline (REML) & 0.484 & 0.969 & 0.257 & 0.624 & 0.128 & 1.000\\
Spatial+ & 0.237 & 0.475 & 0.184 & 1.634 & 0.950 & 0.415\\
gSEM & 0.211 & 0.422 & 0.163 & 1.559 & 0.948 & 0.418\\
DSR (theory) & 0.478 & 0.955 & 0.265 & 0.824 & 0.370 & 0.998\\
DSR (theory, GpGp) & 0.162 & 0.325 & 0.140 & 1.215 & 0.877 & 0.570\\
DSR (theory, no crossfit) & 0.418 & 0.836 & 0.218 & 0.897 & 0.573 & 0.983\\
DSR & \textbf{0.141} & 0.281 & 0.150 & 1.450 & 0.922 & 0.448\\
DSR (no crossfit) & 0.160 & 0.320 & \textbf{0.117} & 1.216 & 0.895 & 0.568\\
DSR (spline, no crossfit) & 0.252 & 0.504 & 0.123 & 0.819 & 0.720 & 0.925\\
DSR (theory, spline, no crossfit) & 0.211 & 0.422 & 0.163 & 1.309 & 0.875 & 0.557\\
DSR (theory, GpGp, smooth) & 0.165 & 0.329 & 0.137 & 1.215 & 0.900 & 0.578\\
DSR (smooth) & 0.152 & 0.305 & 0.150 & 1.461 & 0.917 & 0.465\\

\hline

\end{tabular}

\end{table}

\begin{table}[H]
    \centering
    \small
    \caption{Simulation results for smooth $\bU$ and rough $\bA$. Metrics are bias, relative bias (bias divided by $\beta_0$), mean squared error (MSE), confidence interval (CI) length, coverage (CVG), and power. CI length, coverage, and power are computed with respect to 95\% confidence intervals.}
    \begin{tabular}[t]{lrrrrrrr}
\hline
  & Bias & Rel. Bias & MSE & CI Length & CVG & Power\\
OLS & 0.402 & 0.804 & 0.195 & 0.172 & 0.040 & 1.000\\
LMM & 0.179 & 0.358 & 0.038 & 0.295 & 0.340 & 1.000\\
Spline (GCV) & 0.181 & 0.362 & 0.039 & 0.289 & 0.315 & 1.000\\
Spline (REML) & 0.177 & 0.353 & 0.037 & 0.295 & 0.350 & 1.000\\
Spatial+ & 0.173 & 0.345 & 0.097 & 1.245 & 0.965 & 0.595\\
gSEM & 0.169 & 0.338 & 0.095 & 1.238 & 0.973 & 0.603\\
DSR (theory) & 0.182 & 0.365 & 0.042 & 0.409 & 0.610 & 1.000\\
DSR (theory, GpGp) & 0.027 & 0.054 & 0.042 & 0.756 & 0.935 & 0.762\\
DSR (theory, no crossfit) & 0.236 & 0.472 & 0.072 & 0.466 & 0.488 & 1.000\\
DSR & 0.033 & 0.066 & 0.041 & 0.775 & 0.968 & 0.780\\
DSR (no crossfit) & 0.041 & 0.082 & 0.048 & 0.830 & 0.950 & 0.713\\
DSR (spline, no crossfit) & 0.038 & 0.077 & \textbf{0.036} & 0.713 & 0.950 & 0.830\\
DSR (theory, spline, no crossfit) & 0.169 & 0.338 & 0.095 & 0.958 & 0.885 & 0.770\\
DSR (theory, GpGp, smooth) & \textbf{0.004} & 0.007 & 0.042 & 0.743 & 0.940 & 0.713\\
DSR (smooth) & 0.023 & 0.046 & 0.045 & 0.772 & 0.938 & 0.745\\

\hline

\end{tabular}

    \label{tab:sim_smoothU_roughA}
\end{table}

\normalsize

\begin{table}[H]
    \centering
    \small
    \caption{Simulation results for smooth $\bU$ and smooth $\bA$. Metrics are bias, relative bias (bias divided by $\beta_0$), mean squared error (MSE), confidence interval (CI) length, coverage (CVG), and power. CI length, coverage, and power are computed with respect to 95\% confidence intervals.}
    \begin{tabular}[t]{lrrrrrrr}
\hline
  & Bias & Rel. Bias & MSE & CI Length & CVG & Power\\
OLS & 0.498 & 0.997 & 0.291 & 0.174 & 0.028 & 1.000\\
LMM & 0.434 & 0.869 & 0.206 & 0.461 & 0.065 & 1.000\\
Spline (GCV) & 0.435 & 0.870 & 0.208 & 0.463 & 0.075 & 1.000\\
Spline (REML) & 0.431 & 0.862 & 0.204 & 0.477 & 0.088 & 1.000\\
Spatial+ & 0.210 & 0.419 & 0.163 & 1.619 & 0.950 & 0.398\\
gSEM & 0.194 & 0.388 & 0.151 & 1.552 & 0.950 & 0.412\\
DSR (theory) & 0.431 & 0.862 & 0.216 & 0.790 & 0.448 & 1.000\\
DSR (theory, GpGp) & 0.046 & 0.093 & 0.098 & 1.174 & 0.938 & 0.462\\
DSR (theory, no crossfit) & 0.386 & 0.772 & 0.189 & 0.884 & 0.618 & 0.983\\
DSR & 0.018 & 0.036 & 0.121 & 1.458 & 0.963 & 0.345\\
DSR (no crossfit) & 0.047 & 0.094 & 0.093 & 1.295 & 0.968 & 0.415\\
DSR (spline, no crossfit) & 0.180 & 0.359 & \textbf{0.080} & 0.840 & 0.828 & 0.858\\
DSR (theory, spline, no crossfit) & 0.194 & 0.388 & 0.151 & 1.320 & 0.887 & 0.545\\
DSR (theory, GpGp, smooth) & 0.048 & 0.096 & 0.095 & 1.175 & 0.930 & 0.452\\
DSR (smooth) & \textbf{0.016} & 0.032 & 0.114 & 1.448 & 0.970 & 0.332\\

\hline

\end{tabular}

\end{table}

\begin{table}[H]
    \centering
    \small
    \caption{Simulation results for middle-out heteroskedastic errors. Metrics are bias, relative bias (bias divided by $\beta_0$), mean squared error (MSE), confidence interval (CI) length, coverage (CVG), and power. CI length, coverage, and power are computed with respect to 95\% confidence intervals.}
    \begin{tabular}[t]{lrrrrrrr}
\hline
  & Bias & Rel. Bias & MSE & CI Length & CVG & Power\\
OLS & 0.165 & 0.329 & 0.036 & 0.106 & 0.102 & 1.000\\
LMM & 0.151 & 0.301 & \textbf{0.027} & 0.236 & 0.310 & 1.000\\
Spline (GCV) & 0.151 & 0.302 & 0.028 & 0.266 & 0.428 & 1.000\\
Spline (REML) & 0.152 & 0.304 & 0.028 & 0.252 & 0.380 & 1.000\\
Spatial+ & 0.162 & 0.325 & 0.087 & 1.139 & 0.948 & 0.632\\
gSEM & 0.149 & 0.298 & 0.081 & 1.089 & 0.950 & 0.637\\
DSR (theory) & 0.220 & 0.441 & 0.063 & 0.550 & 0.667 & 1.000\\
DSR (theory, GpGp) & 0.020 & 0.039 & 0.049 & 0.824 & 0.945 & 0.680\\
DSR (theory, no crossfit) & 0.210 & 0.420 & 0.067 & 0.621 & 0.738 & 0.995\\
DSR & 0.010 & 0.020 & 0.061 & 0.970 & 0.975 & 0.565\\
DSR (no crossfit) & 0.022 & 0.045 & 0.044 & 0.847 & 0.958 & 0.665\\
DSR (spline, no crossfit) & 0.062 & 0.123 & \textbf{0.027} & 0.585 & 0.938 & 0.940\\
DSR (theory, spline, no crossfit) & 0.149 & 0.298 & 0.081 & 0.934 & 0.900 & 0.760\\
DSR (theory, GpGp, smooth) & 0.020 & 0.041 & 0.049 & 0.823 & 0.940 & 0.680\\
DSR (smooth) & \textbf{0.007} & 0.014 & 0.063 & 0.980 & 0.965 & 0.555\\

\hline

\end{tabular}

    \label{tab:sim_middleout}
\end{table}

\begin{table}[H]
    \centering
    \small
    \caption{Simulation results for cubed confounder. Metrics are bias, relative bias (bias divided by $\beta_0$), mean squared error (MSE), confidence interval (CI) length, coverage (CVG), and power. CI length, coverage, and power are computed with respect to 95\% confidence intervals.}
    \begin{tabular}[t]{lrrrrrrr}
\hline
  & Bias & Rel. Bias & MSE & CI Length & CVG & Power\\
OLS & 1.449 & 2.899 & 3.164 & 0.426 & 0.025 & 1.000\\
LMM & 0.465 & 0.929 & 0.292 & 1.029 & 0.557 & 0.925\\
Spline (GCV) & 0.558 & 1.115 & 0.398 & 1.021 & 0.455 & 0.970\\
Spline (REML) & 0.466 & 0.931 & 0.294 & 1.069 & 0.598 & 0.927\\
Spatial+ & 0.114 & 0.228 & 0.141 & 1.638 & 0.980 & 0.285\\
gSEM & 0.077 & 0.155 & 0.119 & 1.552 & 0.978 & 0.282\\
DSR (theory) & 0.935 & 1.869 & 1.175 & 1.163 & 0.192 & 0.990\\
DSR (theory, GpGp) & 0.053 & 0.107 & 0.136 & 1.260 & 0.915 & 0.402\\
DSR (theory, no crossfit) & 0.749 & 1.498 & 0.777 & 1.163 & 0.338 & 0.983\\
DSR & 0.009 & 0.018 & 0.150 & 1.602 & 0.953 & 0.280\\
DSR (no crossfit) & 0.089 & 0.179 & \textbf{0.112} & 1.203 & 0.900 & 0.488\\
DSR (spline, no crossfit) & 0.263 & 0.526 & 0.158 & 0.798 & 0.680 & 0.882\\
DSR (theory, spline, no crossfit) & 0.077 & 0.155 & 0.119 & 1.254 & 0.925 & 0.448\\
DSR (theory, GpGp, smooth) & 0.059 & 0.117 & 0.134 & 1.260 & 0.912 & 0.438\\
DSR (smooth) & \textbf{0.001} & 0.002 & 0.159 & 1.636 & 0.955 & 0.288\\

\hline

\end{tabular}

    \label{tab:sim_cubed}
\end{table}

\begin{table}[H]
    \centering
    \small
    \caption{Simulation results for gamma-distributed errors. Metrics are bias, relative bias (bias divided by $\beta_0$), mean squared error (MSE), confidence interval (CI) length, coverage (CVG), and power. CI length, coverage, and power are computed with respect to 95\% confidence intervals.}
    \begin{tabular}[t]{lrrrrrrr}
\hline
  & Bias & Rel. Bias & MSE & CI Length & CVG & Power\\
OLS & 0.416 & 0.832 & 0.210 & 0.166 & 0.055 & 1.000\\
LMM & 0.367 & 0.733 & 0.151 & 0.429 & 0.108 & 1.000\\
Spline (GCV) & 0.364 & 0.728 & 0.151 & 0.452 & 0.148 & 1.000\\
Spline (REML) & 0.364 & 0.729 & 0.151 & 0.450 & 0.142 & 1.000\\
Spatial+ & 0.199 & 0.398 & 0.158 & 1.619 & 0.963 & 0.385\\
gSEM & 0.184 & 0.368 & 0.148 & 1.552 & 0.960 & 0.400\\
DSR (theory) & 0.390 & 0.780 & 0.183 & 0.790 & 0.500 & 0.995\\
DSR (theory, GpGp) & 0.034 & 0.068 & 0.094 & 1.171 & 0.953 & 0.448\\
DSR (theory, no crossfit) & 0.360 & 0.720 & 0.171 & 0.888 & 0.665 & 0.975\\
DSR & 0.008 & 0.016 & 0.119 & 1.435 & 0.965 & 0.335\\
DSR (no crossfit) & 0.038 & 0.075 & 0.094 & 1.278 & 0.965 & 0.410\\
DSR (spline, no crossfit) & 0.150 & 0.299 & \textbf{0.072} & 0.834 & 0.858 & 0.835\\
DSR (theory, spline, no crossfit) & 0.184 & 0.368 & 0.148 & 1.323 & 0.890 & 0.530\\
DSR (theory, GpGp, smooth) & 0.035 & 0.070 & 0.093 & 1.173 & 0.943 & 0.450\\
DSR (smooth) & \textbf{-0.002} & -0.004 & 0.118 & 1.436 & 0.973 & 0.340\\

\hline

\end{tabular}

    \label{tab:sim_gamma}
\end{table}

\begin{table}[H]
    \centering
    \small
    \caption{Simulation results for east-west heteroskedastic errors. Metrics are bias, relative bias (bias divided by $\beta_0$), mean squared error (MSE), confidence interval (CI) length, coverage (CVG), and power. CI length, coverage, and power are computed with respect to 95\% confidence intervals.}
    \begin{tabular}[t]{lrrrrrrr}
\hline
  & Bias & Rel. Bias & MSE & CI Length & CVG & Power\\
OLS & 0.250 & 0.499 & 0.078 & 0.150 & 0.072 & 1.000\\
LMM & 0.230 & 0.461 & 0.062 & 0.331 & 0.230 & 1.000\\
Spline (GCV) & 0.230 & 0.461 & 0.063 & 0.369 & 0.320 & 1.000\\
Spline (REML) & 0.232 & 0.464 & 0.063 & 0.351 & 0.282 & 1.000\\
Spatial+ & 0.195 & 0.389 & 0.156 & 1.618 & 0.960 & 0.382\\
gSEM & 0.182 & 0.365 & 0.147 & 1.552 & 0.948 & 0.388\\
DSR (theory) & 0.350 & 0.700 & 0.151 & 0.783 & 0.603 & 0.993\\
DSR (theory, GpGp) & 0.026 & 0.052 & 0.094 & 1.166 & 0.938 & 0.445\\
DSR (theory, no crossfit) & 0.331 & 0.662 & 0.156 & 0.887 & 0.703 & 0.970\\
DSR & \textbf{0.006} & 0.013 & 0.109 & 1.368 & 0.968 & 0.352\\
DSR (no crossfit) & 0.024 & 0.048 & 0.088 & 1.232 & 0.965 & 0.420\\
DSR (spline, no crossfit) & 0.092 & 0.184 & \textbf{0.053} & 0.832 & 0.907 & 0.780\\
DSR (theory, spline, no crossfit) & 0.182 & 0.365 & 0.147 & 1.333 & 0.895 & 0.520\\
DSR (theory, GpGp, smooth) & 0.023 & 0.046 & 0.089 & 1.166 & 0.960 & 0.432\\
DSR (smooth) & 0.007 & 0.013 & 0.113 & 1.395 & 0.973 & 0.330\\

\hline

\end{tabular}

    \label{tab:sim_eastwest}
\end{table}

\begin{table}[H]
    \centering
    \small
    \caption{Simulation results for $\sigma_A^2 = 1$.  Metrics are bias, relative bias (bias divided by $\beta_0$), mean squared error (MSE), confidence interval (CI) length, coverage (CVG), and power. CI length, coverage, and power are computed with respect to 95\% confidence intervals. Note that Shift (BART) results are included, without variance estimates.}
    \begin{tabular}[t]{lrrrrrrr}
\hline
  & Bias & Rel. Bias & MSE & CI Length & CVG & Power\\
\hline
OLS & 0.232 & 0.464 & 0.065 & 0.121 & 0.070 & 1\\
LMM & 0.025 & 0.050 & 0.002 & 0.124 & 0.855 & 1\\
Spline (GCV) & 0.027 & 0.054 & 0.002 & 0.124 & 0.858 & 1\\
Spline (REML) & 0.024 & 0.049 & 0.002 & 0.124 & 0.870 & 1\\
Spatial+ & 0.020 & 0.041 & 0.001 & 0.161 & 0.978 & 1\\
gSEM & \textbf{0.003} & 0.005 & 0.001 & 0.165 & 0.988 & 1\\
DSR (theory) & 0.012 & 0.023 & 0.001 & 0.124 & 0.930 & 1\\
DSR (theory, GpGp) & 0.008 & 0.015 & 0.001 & 0.124 & 0.940 & 1\\
DSR (theory, no crossfit) & 0.007 & 0.014 & 0.001 & 0.123 & 0.938 & 1\\
DSR & 0.009 & 0.018 & 0.001 & 0.123 & 0.940 & 1\\
DSR (no crossfit) & 0.010 & 0.020 & 0.001 & 0.115 & 0.922 & 1\\
DSR (spline, no crossfit) & 0.009 & 0.018 & 0.001 & 0.115 & 0.920 & 1\\
DSR (theory, spline, no crossfit) & 0.003 & 0.005 & \textbf{0.001} & 0.123 & 0.945 & 1\\
DSR (theory, GpGp, smooth) & 0.008 & 0.016 & 0.001 & 0.124 & 0.948 & 1\\
DSR (smooth) & 0.010 & 0.019 & 0.001 & 0.123 & 0.948 & 1\\

\hline

\end{tabular}

    \label{tab:sim_highvar}
\end{table}

\begin{table}[H]
    \centering
    \small
    \caption{Simulation results for very rough, exponential spatial processes generating $\bU$ and $\bA$. Metrics are bias, relative bias (bias divided by $\beta_0$), mean squared error (MSE), confidence interval (CI) length, coverage (CVG), and power. CI length, coverage, and power are computed with respect to 95\% confidence intervals.}
    \begin{tabular}[t]{lrrrrrrr}
\hline
  & Bias & Rel. Bias & MSE & CI Length & CVG & Power\\
OLS & 0.501 & 1.003 & 0.265 & 0.168 & 0.000 & 1\\
LMM & 0.491 & 0.981 & 0.245 & 0.237 & 0.000 & 1\\
Spline (GCV) & 0.489 & 0.978 & 0.243 & 0.253 & 0.000 & 1\\
Spline (REML) & 0.492 & 0.984 & 0.246 & 0.234 & 0.000 & 1\\
Spatial+ & 0.608 & 1.216 & 0.380 & 0.483 & 0.000 & 1\\
gSEM & 0.582 & 1.165 & 0.350 & 0.481 & 0.002 & 1\\
DSR (theory) & 0.508 & 1.017 & 0.263 & 0.274 & 0.000 & 1\\
DSR (theory, GpGp) & 0.473 & 0.947 & 0.233 & 0.328 & 0.002 & 1\\
DSR (theory, no crossfit) & 0.554 & 1.108 & 0.319 & 0.297 & 0.000 & 1\\
DSR & 0.474 & 0.947 & 0.232 & 0.324 & 0.000 & 1\\
DSR (no crossfit) & 0.478 & 0.956 & 0.236 & 0.282 & 0.002 & 1\\
DSR (spline, no crossfit) & 0.478 & 0.956 & 0.235 & 0.299 & 0.000 & 1\\
DSR (theory, spline, no crossfit) & 0.582 & 1.165 & 0.350 & 0.359 & 0.000 & 1\\
DSR (theory, GpGp, smooth) & \textbf{0.446} & 0.893 & \textbf{0.207} & 0.321 & 0.000 & 1\\
DSR (smooth) & 0.478 & 0.957 & 0.236 & 0.322 & 0.000 & 1\\

\hline

\end{tabular}

    \label{tab:sim_exponential}
\end{table}

\begin{table}[H]
    \centering
    \small
    \caption{Simulation results for spatial locations located on a regular grid. Metrics are bias, relative bias (bias divided by $\beta_0$), mean squared error (MSE), confidence interval (CI) length, coverage (CVG), and power. CI length, coverage, and power are computed with respect to 95\% confidence intervals.}
    \begin{tabular}[t]{lrrrrrrr}
\hline
  & Bias & Rel. Bias & MSE & CI Length & CVG & Power\\
OLS & 0.488 & 0.975 & 0.277 & 0.171 & 0.030 & 1.000\\
LMM & 0.431 & 0.863 & 0.205 & 0.442 & 0.075 & 1.000\\
Spline (GCV) & 0.430 & 0.859 & 0.205 & 0.444 & 0.085 & 1.000\\
Spline (REML) & 0.428 & 0.855 & 0.203 & 0.455 & 0.095 & 1.000\\
Spatial+ & 0.215 & 0.430 & 0.167 & 1.596 & 0.945 & 0.390\\
gSEM & 0.199 & 0.398 & 0.155 & 1.529 & 0.940 & 0.415\\
DSR (theory) & 0.435 & 0.869 & 0.221 & 0.767 & 0.378 & 1.000\\
DSR (theory, GpGp) & 0.046 & 0.092 & 0.095 & 1.152 & 0.943 & 0.455\\
DSR (theory, no crossfit) & 0.407 & 0.815 & 0.212 & 0.866 & 0.560 & 0.993\\
DSR & \textbf{0.006} & 0.013 & 0.116 & 1.446 & 0.965 & 0.325\\
DSR (no crossfit) & 0.025 & 0.050 & 0.101 & 1.373 & 0.968 & 0.328\\
DSR (spline, no crossfit) & 0.180 & 0.360 & \textbf{0.081} & 0.817 & 0.828 & 0.853\\
DSR (theory, spline, no crossfit) & 0.199 & 0.398 & 0.155 & 1.300 & 0.887 & 0.550\\
DSR (theory, GpGp, smooth) & 0.057 & 0.114 & 0.098 & 1.153 & 0.925 & 0.460\\
DSR (smooth) & 0.013 & 0.027 & 0.116 & 1.454 & 0.963 & 0.318\\

\hline

\end{tabular}

    \label{tab:sim_grid}
\end{table}

\begin{table}[H]
    \centering
    \small
    \caption{Simulation results for deterministic latent functions of space $g_0$ and $m_0$ such that $g_0=m_0$. Metrics are bias, relative bias (bias divided by $\beta_0$), mean squared error (MSE), confidence interval (CI) length, coverage (CVG), and power. CI length, coverage, and power are computed with respect to 95\% confidence intervals.}
    \begin{tabular}[t]{lrrrrrrr}
\hline
  & Bias & Rel. Bias & MSE & CI Length & CVG & Power\\
OLS & 0.959 & 1.919 & 0.924 & 0.244 & 0.000 & 1.000\\
LMM & 0.957 & 1.913 & 0.919 & 0.252 & 0.000 & 1.000\\
Spline (GCV) & 0.956 & 1.913 & 0.919 & 0.252 & 0.000 & 1.000\\
Spline (REML) & 0.958 & 1.917 & 0.922 & 0.247 & 0.000 & 1.000\\
Spatial+ & 0.227 & 0.454 & 0.170 & 1.600 & 0.953 & 0.428\\
gSEM & 0.208 & 0.416 & 0.156 & 1.521 & 0.938 & 0.450\\
DSR (theory) & 0.576 & 1.152 & 0.388 & 0.997 & 0.350 & 0.990\\
DSR (theory, GpGp) & 0.045 & 0.090 & 0.109 & 1.209 & 0.920 & 0.458\\
DSR (theory, no crossfit) & 0.467 & 0.935 & 0.300 & 1.087 & 0.588 & 0.932\\
DSR & 0.141 & 0.281 & \textbf{0.102} & 1.227 & 0.925 & 0.568\\
DSR (no crossfit) & 0.301 & 0.602 & 0.159 & 0.967 & 0.745 & 0.858\\
DSR (spline, no crossfit) & 0.371 & 0.743 & 0.181 & 0.874 & 0.603 & 0.963\\
DSR (theory, spline, no crossfit) & 0.208 & 0.416 & 0.156 & 1.304 & 0.892 & 0.580\\
DSR (theory, GpGp, smooth) & \textbf{0.051} & 0.102 & 0.106 & 1.210 & 0.935 & 0.448\\
DSR (smooth) & 0.146 & 0.292 & \textbf{0.102} & 1.225 & 0.938 & 0.580\\

\hline

\end{tabular}

    \label{tab:sim_determ1}
\end{table}

\begin{table}[H]
    \centering
    \small
    \caption{Simulation results for deterministic latent functions of space $g_0$ and $m_0$ such that $g_0\neq m_0$. Metrics are bias, relative bias (bias divided by $\beta_0$), mean squared error (MSE), confidence interval (CI) length, coverage (CVG), and power. CI length, coverage, and power are computed with respect to 95\% confidence intervals.}
    \begin{tabular}[t]{lrrrrrrr}
\hline
  & Bias & Rel. Bias & MSE & CI Length & CVG & Power\\
OLS & 0.987 & 1.974 & 0.979 & 0.271 & 0.000 & 1.000\\
LMM & 0.873 & 1.746 & 0.769 & 0.511 & 0.000 & 1.000\\
Spline (GCV) & 0.844 & 1.688 & 0.722 & 0.541 & 0.000 & 1.000\\
Spline (REML) & 0.849 & 1.698 & 0.729 & 0.538 & 0.000 & 1.000\\
Spatial+ & 0.195 & 0.390 & 0.156 & 1.600 & 0.965 & 0.395\\
gSEM & 0.176 & 0.352 & 0.143 & 1.521 & 0.960 & 0.418\\
DSR (theory) & 0.529 & 1.058 & 0.334 & 1.003 & 0.430 & 0.985\\
DSR (theory, GpGp) & 0.037 & 0.074 & 0.104 & 1.213 & 0.945 & 0.408\\
DSR (theory, no crossfit) & 0.394 & 0.787 & 0.226 & 1.087 & 0.713 & 0.910\\
DSR & 0.085 & 0.171 & \textbf{0.099} & 1.254 & 0.945 & 0.458\\
DSR (no crossfit) & 0.238 & 0.477 & 0.120 & 0.993 & 0.812 & 0.818\\
DSR (spline, no crossfit) & 0.333 & 0.666 & 0.161 & 0.836 & 0.618 & 0.958\\
DSR (theory, spline, no crossfit) & 0.176 & 0.352 & 0.143 & 1.298 & 0.892 & 0.542\\
DSR (theory, GpGp, smooth) & \textbf{0.043} & 0.087 & 0.104 & 1.213 & 0.930 & 0.425\\
DSR (smooth) & 0.087 & 0.173 & 0.101 & 1.252 & 0.948 & 0.465\\

\hline

\end{tabular}

    \label{tab:sim_determ2}
\end{table}

\normalsize

\section{Proof of Theorem 1}

Theorem 1 follows from an extension of Theorem 4.1 from \citep{dml}, and Theorems 3.3 and 3.6 from \citep{eberts2013optimal}. This section extends Theorem 4.1 from \citep{dml} to the case of a vector treatment variable, and then verifies that Assumptions A1-A6 satisfy the necessary conditions to apply these results.

\subsection{Proof of DML with partially linear model}

In this section we extend Theorem 4.1 from \citep{dml}, which analyzes the partially linear model with a scalar treatment, to a vector treatment. The extension essentially follows the proof of Theorem 4.1 from \citep{dml} with slight changes.

We use ``DML2", which is Definition 3.2 from \cite{dml}, where rather than aggregating K different estimates from K different folds, cross-fitting on the K folds is performed followed by estimation of $\bbeta_0$ using the combined cross-fitted estimates.

The assumed model is:
\begin{align}
\begin{split}\label{eqn:appendix_mod}
    \bY &= \bX\bbeta + g_0(\bS) + \bU \\
    \bX_j &= m_{0j}(\bS) + \bV_j
    \end{split}
\end{align}
with notation and definitions as in the main paper. Nuisance parameters and estimates $\eta$ are assumed to be in $T$, a convex subset of some normed vector space.

The ``practical" DSR estimator uses the score function:
\begin{align}\label{eqn:score1}
    \psi(\bW; \bbeta, \eta) := \{ Y - g(\bS) - \bX^T\bbeta\}(\bX - m(\bS))
\end{align}
However, for theory, we focus on the score function:
\begin{align}\label{eqn:score2}
    \psi(\bW; \bbeta, \eta) := \{ Y - h(\bS) - (\bX - m(\bS))^T\bbeta\}(\bX - m(\bS))
\end{align}

Assumption \ref{assn:plr} consists of those of Assumption 4.1 from \cite{dml} but expanded to encompass the case $p>1$, the length of $\bbeta_0$, and with some other slight changes for theoretical convenience. Let $\{\delta_n\}$ and $\{\Delta_n\}$ be sequences of positive constants converging to 0. Let $c, C, $ and $q$ be fixed strictly positive constants such that $q>4$, and let $K\geq 2$ be a fixed integer. For any $\eta=(\ell_1, \ell_2, ..., \ell_m)$ for any positive integer $m$ such that $\ell_1$, ..., $\ell_m$ are functions mapping $\mathcal{S}$ to $\mathbb{R}$, denote $\|\eta\|_{P,q} = \max_{1, ..., m}\{\|\ell_1\|_{P,q}, ..., \|\ell_m\|_{P,q}\}$.
\begin{assumption}[Regularity Conditions for partially linear regression model]\label{assn:plr}
    Let $\mathcal{P}$ be the collection of probability laws $P$ for $\bW=(Y, \bX, \bS)$ such that
    \begin{enumerate}
        \item[a)] Model (\ref{eqn:appendix_mod}) holds,
        \item[b)] $\|V_1\|_{P,q} , ... , \|V_p\|_{P,q} \leq C$ and $\|\bbeta_0\|_{\infty} \leq C$,
        \item[c)] The eigenvalues of the matrix $E[U^2\bV\bV^T]$ are greater than or equal to $ c^2$ and the matrix $E[\bV\bV^T]$ has singular values at least $ c$ and no greater than $C$,
        \item[d)] $\| E[U^2|\bS]\|_{P, \infty} \leq C$ and $\| E[V_j^2|\bS]\|_{P, \infty} \leq C$ for $j=1,...,p$,
        \item[e)] Given a random subset $I$ of $[n]$ of size $n/K$, the nuisance parameter estimator $\widehat{\eta}_0=\widehat{\eta}_0((\bW_i)_{i\in I^C})$ obey the following conditions for all $n/K \geq 1$: with $P$-probability no less than $1 - \Delta_n$, 
        \[  \|\widehat{\eta}_0 - \eta_0\|_{P,\infty} \leq C, \|\widehat{\eta}_0 - \eta_0\|_{P,2} \leq \delta_n,\]
        and for the score $\psi$ in (\ref{eqn:score2}), where $\widehat{\eta} = (\widehat{\ell}_0,\widehat{m}_{10}, ..., \widehat{m}_{p0}) $,
            \[ \|\widehat{m}_{j0}-m_{j0}\|_{P,2} \times (\|\widehat{m}_{j0}-m_{j0}\|_{P,2} + \|\widehat{\ell}_{0}-\ell_{0}\|_{P,2} ) \leq \delta_nn^{-1/2}\text{  for  } j=1,...,p.\]
    \end{enumerate}
\end{assumption}

The following lemma is Theorem 4.1 from \cite{dml} but expanded to the case $p>1$.
\begin{lemma}[DML inference in the partially linear regression model with $p>1$]\label{theorem:plr}
    Suppose that Assumption \ref{assn:plr} holds. Then the DML2 estimator constructed in Definition 3.2 of \cite{dml} using the score (\ref{eqn:score2}) obeys
    \[  \boldsymbol{\Sigma}^{-1/2}\sqrt{n}(\widehat{\bbeta}_0 - \bbeta_0) \stackrel{D}{\rightarrow} N(\bold{0}, \bold{I}_p),\]
    uniformly over $P \in \mathcal{P}$, where $\boldsymbol{\Sigma} = [E(\bV\bV^T)]^{-1}E(U^2\bV\bV^T)[E(\bV\bV^T)]^{-1}$. The result continues to hold if $\boldsymbol{\Sigma}$ is replaced by $\widehat{\boldsymbol{\Sigma}}$ from Theorem 3.2 from \cite{dml}.
\end{lemma}

The proof for Lemma \ref{theorem:plr} follows the proof of Theorem 4.1 in \cite{dml} closely, only needing to add steps to deal with the vector-valued $\psi$ and matrix-valued $\psi^a$, and changing some assumptions slightly for convenience. The proof verifies Assumptions 3.1 and 3.2 from \cite{dml}, from which the conclusion of Lemma \ref{theorem:plr} follows from Theorems 3.1 and 3.2 from \cite{dml}.

\begin{proof}[Proof of Lemma \ref{theorem:plr}]
    Observe that the score (\ref{eqn:score2}) is linear in $\bbeta$:
    \begin{align*}
        \psi(\bW; \bbeta, \eta) &= \{ Y - \ell(\bS) - (\bX - \bold{m}(\bS))^T\bbeta\}(\bX - \bold{m}(\bS)) = \psi^a(\bW; \eta)\bbeta + \psi^b(\bW; \eta), \\
        \psi^a(\bW; \eta)&= -(\bX - \bold{m}(\bS))(\bX - \bold{m}(\bS))^T, \quad \psi^b(\bW; \eta) = (Y - \ell(\bS))(\bX - \bold{m}(\bS))
    \end{align*}
    Therefore, it is sufficient to verify Assumptions 3.1 and 3.2 from \cite{dml}. Let $\mathcal{T}_n$ be the set of all $\eta = (\ell, m_1, ..., m_p)$ consisting of $P-$square-integrable functions $\ell, m_1, ..., m_p$ such that 
    \begin{align*}
        \|\widehat{\eta}_0 - \eta_0\|_{P,q} \leq C, \|\widehat{\eta}_0 - \eta_0\|_{P,2} &\leq \delta_n, \\
        \|\widehat{m}_{j0}-m_{j0}\|_{P,2} \times (\|\widehat{m}_{j0}-m_{j0}\|_{P,2} + \|\widehat{\ell}_{0}-\ell_{0}\|_{P,2} ) &\leq \delta_nn^{-1/2}.
    \end{align*}
    We replace the constant $q$ and the sequence $\{\delta_n\}$ in Assumptions 3.1 and 3.2 from \cite{dml} by $q/2$ and $\{\delta_n'\}$, with $\delta'_n=(2p^2C + \sqrt{C}p + \sqrt{pC} + pC+pC\sqrt{pC} + \sqrt{p}pC^2+4\sqrt{p})(\delta_n \vee n^{-[(1-4/q)\wedge(1/2)]})$ for all $n$. As in \cite{dml}, we use five steps.

    \textbf{Step 1.} Verify Neyman orthogonality. Note that $E[\psi(\bW; \bbeta_0, \eta_0)]=0$ by the definitions of $\bbeta_0, \eta_0$. For any $\eta\in\mathcal{T}_n$, the Gateaux derivative in the direction $\eta - \eta_0$ is given by, for $r=0$, (see derivation in Step 5):
    \begin{align*}
        \partial_\eta E[\psi(\bW; \bbeta_0, \eta_0)][\eta-\eta_0] &= -E[(Y - \ell_0(\bS))(\bmm(\bS) - \bmm_0(\bS))] - E[(\ell(\bS) - \ell_0(\bS))(\bD - \bmm_0(\bS))] \\&+ E[(\bD - \bmm_0(\bS))(\bmm(\bS)-\bmm_0(\bS))^T\bbeta_0] \\
        &+ E[(\bmm(\bS)-\bmm_0(\bS))(\bD - \bmm_0(\bS))^T\bbeta_0].
    \end{align*}
    By the law of iterated expectation, and since  $\bV=\bX-\bmm_0(\bS)$ and $U = Y-\ell_0(\bS)$:
    \begin{align*}
        \partial_\eta E[\psi(\bW; \bbeta_0, \eta_0)][\eta-\eta_0] &= E\{E[U\times(\bmm(\bS) - \bmm_0(\bS))|\bX, \bS]\} \\
        &- E\{E[(\ell(\bS) - \ell_0(\bS))\bV|\bS]\} \\
        &+ E\{E[(\bmm(\bS) - \bmm_0(\bS))\bV^T\bbeta_0|\bS]\} \\
        &+ E\{E[\bV(\bmm(\bS) - \bmm_0(\bS))^T\bbeta_0|\bS]\},
    \end{align*}
    which is equal to $0$ since $E[\bV|\bS]=\bold{0}$ and $E[U|\bX, \bS]=0$. This gives Assumption 3.1(d) from \cite{dml} with $\lambda_n=0$.

    \textbf{Step 2.} By definition,
    \begin{align*}
        J_0 &= E[\psi^a(\bW; \eta_0)] \\
        &= E[-(\bX - \bmm_0(\bS))(\bX - \bmm_0(\bS))^T] \\
        &= E[-\bV\bV^T]
    \end{align*}
    By assumption \ref{assn:plr}(c), the singular values of this matrix are between $c_0$ and $C$. This satisfies Assumption 3.1 (e) from \cite{dml}. Since the map $\eta \mapsto E[\psi(\bW; \bbeta_0, \eta)]$ is twice Gateux-differentiable in $T$, this completes verification of Assumption 3.1 from \cite{dml}.

    \textbf{Step 3.} Assumption 3.2(a) from \cite{dml} holds by construction of $\mathcal{T}_n$ and Assumption \ref{assn:plr}(e). In addition, $\psi(\bW; \bbeta_0, \eta_0) = U\bV$, so the eigenvalues of the matrix
    \begin{align*}
        E[\psi(\bW; \bbeta_0, \eta_0)\psi(\bW; \bbeta_0, \eta_0)^T] = E[U^2\bV\bV^T]
    \end{align*}
    are bounded below by $c_0$ by Assumption \ref{assn:plr}(c). This verifies Assumption 3.2(d) from \cite{dml}.

    \textbf{Step 4.} Next, we verify Assumption 3.2(b) from \cite{dml}. For any $\eta = (\ell, \bmm) \in \mathcal{T}_n$,
    \begin{align*}
        m'_n &= E[\|\psi^a(\bS; \eta)\|^{q/2}]^{2/q} \\
        &= \bigl\| \|(\bX - \bmm(\bS))(\bX - \bmm(\bS))^T\| \bigr\|_{P, q/2} \\
        &\leq  \bigl\| \|(\bX - \bmm_0(\bS))(\bX - \bmm_0(\bS))^T\| + \|(\bX - \bmm_0(\bS))(\bmm_0(\bS) - \bmm(\bS))^T\| \\
        &+ \|(\bmm_0(\bS)-\bmm(\bS))(\bX-\bmm_0(\bS))^T\| + \|(\bmm_0(\bS)-\bmm(\bS))(\bmm_0(\bS) - \bmm(\bS))^T\| \bigr\|_{P, q/2} \\
        &= \bigl\| \|\bV\bV^T\| + 2\|\bV(\bmm_0(\bS) - \bmm(\bS))^T\|  + \|(\bmm_0(\bS) - \bmm(\bS))(\bmm_0(\bS) - \bmm(\bS))^T\| \bigr\|_{P, q/2}\\
        &\leq \bigl\| \|\bV\bV^T\| \bigr\|_{P, q/2} + 2\bigl\| \|\bV(\bmm_0(\bS) - \bmm(\bS))^T\|\bigr\|_{P, q/2} + \bigl\|\|(\bmm_0(\bS) - \bmm(\bS))(\bmm_0(\bS) - \bmm(\bS))^T\|\bigr\|_{P, q/2}\\&= \bigl\| \|\bV\|^2 \bigr\|_{P, q/2} + 2\bigl\| \|\bV\|\times\|(\bmm_0(\bS) - \bmm(\bS))\|\bigr\|_{P, q/2} + \bigl\|\|(\bmm_0(\bS) - \bmm(\bS))\|^2\bigr\|_{P, q/2}\\
        &\leq \bigl\| \|\bV\| \bigr\|_{P,q}^2 + 2\bigl\| \|\bV\|\bigr\|_{P,q}\times \bigl\|\|\bmm_0(\bS) - \bmm(\bS)\|\bigr\|_{P,q} + \bigl\|\|\bmm_0(\bS) - \bmm(\bS)\|\bigr\|_{P,q}^2,
    \end{align*}
    by the triangle inequality for the $\|\cdot\|$-norm, the triangle inequality for the $\|\cdot\|_{P, q/2}$-norm, the fact that for vectors $\bu$ and $\bv$, $\|\bu\bv^T\|=\|\bu\|\|\bv\|$, and by the Cauchy-Schwarz inequality. Note that for a random vector $\bZ$ with $p$ elements, if $\|Z_j\|_{P,q}<C$ for $j=1, ..., p$, then $\| \|\bZ\| \|_{P,q} \leq \| \|\bZ\|_1 \|_{P,q} \leq \|Z_1\|_{P,q} + ... + \|Z_p\|_{P,q} \leq pC$, by the fact that $\|\bZ\| \leq \|\bZ\|_1$ and monotonicity of expectation, then the triangle inequality, then the assumption that $\|Z_j\|_{P,q}<C$ for $j=1,...,p$. By assumption $ \|V_1\|_{P,q} + ... + \|V_p\|_{P,q} \leq C$. Also, for $\eta \in \mathcal{T}_n, \|\eta_0 - \eta\|_{P,q}< C $. Therefore we have that:
    \begin{align*}
        m'_n & \leq \| \|\bV\| \|_{P,q}^2 + 2\| \|\bV\|\|_{P,q}\times \|\|\bmm_0(\bS) - \bmm(\bS)\|\|_{P,q} + \|\|\bmm_0(\bS) - \bmm(\bS)\|\|_{P,q}^2 \\
        & \leq (pC)^2 + 2(pC)(pC) + (pC)^2= 4p^2C^2
    \end{align*}
    bounding $m'_n$ as desired.

    Next we address $m_n$. Note that by assumption, $\|\bbeta_0\| < \sqrt{p}C$ since $\|\beta_0\|_{\infty} < C$. Note also that $Y - \ell_0(\bS) = U + \bV^T\bbeta_0$.
    \begin{align*}
        m_n &= (E[\|\psi(\bW; \bbeta_0, \eta)\|^{q/2}])^{2/q} \\
        &= \bigl\| \| \psi(\bW; \bbeta_0, \eta)\| \bigr\|_{P, q/2} \\
        &= \bigl\| \|U\bV + U(\bmm_0(\bS) - \bmm(\bS)) + (\ell_0(\bS) - \ell(\bS))\bV \\
        &\quad \quad + (\ell_0(\bS) - \ell(\bS))(\bmm_0(\bS) - \bmm(\bS)) \\
        &\quad \quad - (\bmm_0(\bS) - \bmm(\bS))^T\bbeta_0\bV - (\bmm_0(\bS) - \bmm(\bS))(\bmm_0(\bS) - \bmm(\bS))^T\bbeta_0 \| \bigr\|_{P, q/2} \\
        &\leq \bigl\| \|U\bV\| + \|U(\bmm_0(\bS) - \bmm(\bS))\| + \|(\ell_0(\bS) - \ell(\bS))\bV\| + \|(\ell_0(\bS) - \ell(\bS))(\bmm_0(\bS) - \bmm(\bS))\| \\
        &\quad \quad + \|(\bmm_0(\bS) - \bmm(\bS))^T\bbeta_0\bV\| + \|(\bmm_0(\bS) - \bmm(\bS))(\bmm_0(\bS)-\bmm(\bS))^T\bbeta_0 \| \bigr\|_{P, q/2} \\
        &\leq \bigl\| \|U\bV\|\bigr\|_{P, q/2} + \bigl\|\|U(\bmm_0(\bS) - \bmm(\bS))\|\bigr\|_{P, q/2} + \bigl\|\|(\ell_0(\bS) - \ell(\bS))\bV\|\bigr\|_{P, q/2} \\
        & \quad \quad +\bigl\|\|(\ell_0(\bS) - \ell(\bS))(\bmm_0(\bS) - \bmm(\bS))\|\bigr\|_{P, q/2} \\
        &\quad \quad + \bigl\|\|\bV(\bmm_0(\bS) - \bmm(\bS))^T\bbeta_0\|\bigr\|_{P, q/2} + \bigl\|\|(\bmm_0(\bS) - \bmm(\bS))(\bmm_0(\bS)-\bmm(\bS))^T\bbeta_0 \| \bigr\|_{P, q/2} \\
        &\leq \bigl\| U\bigr\|_{P, q} \bigl\| \|\bV\|\bigr\|_{P, q} + \bigl\| U\bigr\|_{P, q}\bigl\| \|\bmm_0(\bS) - \bmm(\bS)\|\bigr\|_{P, q} +  \bigl\|\ell_0(\bS) - \ell(\bS)\bigr\|_{P, q}  \bigl\| \|\bV\|\bigr\|_{P, q} \\
        &\quad \quad + \bigl\| \ell_0(\bS) - \ell(\bS)\bigr\|_{P, q} \bigl\|\|\bmm_0(\bS) - \bmm(\bS)\|\bigr\|_{P, q} \\
        &\quad \quad + \bigl\|\|\bV\|\bigr\|_{P, q} \bigl\|\|\bmm_0(\bS) - \bmm(\bS)\|\bigr\|_{P, q} \|\bbeta_0\| + \bigl\|\|\bmm_0(\bS) - \bmm(\bS)\|\bigr\|_{P, q}^2 \|\bbeta_0\|  \\
        & \leq pC^2 + pC^2 + pC^2 + pC^2 + p^2C^3 + p^2C^3 \\
        &= 4pC^2 + 2p^2C^3,
    \end{align*}
    bounding $m_n$ as desired. The first two inequalities are due to the triangle inequality; 
    the third inequality due to Cauchy-Schwartz; and the final inequality by assumed bounds on the quantities in the previous step and the bounds established in the derivation of the bound on $m_n'$ above.

    \textbf{Step 5.} Finally, we verify the conditions of Assumption 3.2(c) from \cite{dml}. Starting with $r_n = \sup_{\eta \in \mathcal{T}_n}\|E[\psi^a(\bW; \eta)] - E[\psi^a(\bW; \eta_0)]\|$,

    \begin{align*}
        \|E[\psi^a(\bW; \eta)] - E[\psi^a(\bW; \eta_0)]\| &= \| E[\psi^a(\bW; \eta) - \psi^a(\bW; \eta_0)] \|\\
        &= \| E[-(\bX - \bmm(S))(\bX - \bmm(S))^T + \bV\bV^T] \| \\
        &\leq \|E[\bV(\bmm_0(\bS) - \bmm(\bS))^T]\| + \| E[(\bmm_0(\bS) - \bmm(\bS))\bV^T] \| \\
        &\quad \quad + \|E[(\bmm_0(\bS) - \bmm(\bS))(\bmm_0(\bS) - \bmm(\bS))^T]\| \\
        &\leq 2E\bigl[ \|\bV\| \cdot \|(\bmm_0(\bS) - \bmm(\bS))\| \bigr] + E\bigl[ \|(\bmm_0(\bS) - \bmm(\bS))\|^2 \bigr] \\
        &\leq 2\sqrt{E\bigl[\|\bV\|^2\bigr]E\bigl[\|\bmm_0(\bS) - \bmm(\bS)\|^2\bigr]} + E\bigl[\|\bmm_0(\bS) - \bmm(\bS)\|^2 \bigr] \\
        &= 2 \bigl\|\|\bV\| \bigr\|_{P,2} \cdot \bigl\| \|\bmm_0(\bS) - \bmm(\bS)\| \bigr\|_{P,2} + \bigl\|\bmm_0(\bS) - \bmm(\bS)\| \bigr\|^2_{P,2} \\
        &\leq (pC)(p\delta_n) + (p^2 C \delta_n) \\
        &\leq \delta_n'
    \end{align*}
    Where the first inequality is by the triangle inequality, the second inequality is by Jensen's inequality, the third inequality is by Cauchy-Schwarz, and the following inequality by the assumption that $\|V_j\|_{P,2} < C$ and $\|m_{0j} - m_j\|_{P,2} \leq \delta_n$ and $\|m_{0j} - m_j\|_{P,q} \leq C$ for $\eta \in \mathcal{T}_n$ for $j=1,...,p$ and $q > 4$ (since $\|f\|_{P,q_1} \leq \|f\|_{P, q_2}$ for $0 < q_1 < q_2 < \infty$ by Jensen's inequality with $\phi(x)=|x|^{q_2/q_1}$). This establishes the bound on $r_n$. 
    
    Next we establish the bound on $r_n' = \sup_{\eta \in \mathcal{T}_n}(E[\|\psi(\bW; \bbeta_0, \eta) - \psi(\bW; \bbeta_0, \eta_0)\|^2])^{1/2}$:

    \begin{align*}
        \bigl(E[\|\psi(\bW; \bbeta_0, \eta) - \psi(\bW; \bbeta_0, \eta_0)\|^2] \bigr)^{1/2} &= \bigl\| \|\psi(\bW; \bbeta_0, \eta) - \psi(\bW; \bbeta_0, \eta_0)  \bigr\|_{P,2} \\
        &= \bigl\| \|\{Y - \ell(\bS) - (\bX - \bmm(\bS))^T\bbeta_0\}(\bX - \bmm(\bS))- \\
        &\quad \quad \{Y - \ell_0(\bS) - (\bX - \bmm_0(\bS))^T\bbeta_0\}(\bX - \bmm_0(\bS)) \| \bigr\|_{P,2} \\
        &\leq \bigl\| \| U(\bmm_0(\bS) - \bmm_0(\bS))\|  + \|(\ell_0(\bS) - \ell(\bS))\bV\| \\
        &\quad \quad + \|(\ell_0(\bS) - \ell(\bS))(\bmm_0(\bS) - \bmm_0(\bS))\|\\
        &\quad \quad + \|\bV(\bmm_0(\bS) - \bmm(\bS))^T\bbeta_0\| \\
        &\quad \quad+ \|(\bmm_0(\bS) - \bmm(\bS))(\bmm_0(\bS) - \bmm(\bS))^T\bbeta_0 \| \bigr\|_{P,2} \\
        &\leq \bigl\| |U|\cdot\| (\bmm_0(\bS) - \bmm_0(\bS))\|\bigr\|_{P,2}  + \bigl\||\ell_0(\bS) - \ell(\bS)|\cdot\|\bV\|\bigr\|_{P,2} \\
        &\quad \quad + \bigl\||\ell_0(\bS) - \ell(\bS)|\cdot\|(\bmm_0(\bS) - \bmm_0(\bS))\|\bigr\|_{P,2} \\
        &\quad \quad + \bigl\|\|\bV(\bmm_0(\bS) - \bmm(\bS))^T\bbeta_0\|\bigr\|_{P,2} \\
        &\quad \quad+ \bigl\|\|(\bmm_0(\bS) - \bmm(\bS))(\bmm_0(\bS) - \bmm(\bS))^T\bbeta_0 \| \bigr\|_{P,2} \\
        &\leq \bigl\| |U|\cdot\| (\bmm_0(\bS) - \bmm_0(\bS))\|\bigr\|_{P,2}  + \bigl\||\ell_0(\bS) - \ell(\bS)|\cdot\|\bV\|\bigr\|_{P,2} \\
        &\quad \quad + \bigl\||\ell_0(\bS) - \ell(\bS)|\cdot\|(\bmm_0(\bS) - \bmm_0(\bS))\|\bigr\|_{P,2} \\
        &\quad \quad + \bigl\|\|\bV\|\cdot\|\bmm_0(\bS) - \bmm(\bS)\|\bigr\|_{P,2}\|\bbeta_0\| \\
        &\quad \quad+ \bigl\|\|\bmm_0(\bS) - \bmm(\bS)\|^2 \bigr\|_{P,2}\|\bbeta_0 \| \\
        &\leq \sqrt{C}p\delta_n + \sqrt{pC}\delta_n + pC\delta_n + \sqrt{pC}p\delta_nC + \sqrt{p}Cp\delta_nC \\
        &= (\sqrt{C}p + \sqrt{pC}+pC + \sqrt{pC}pC + \sqrt{p}C^2p)\delta_n \\
        &\leq \delta_n'
    \end{align*}
    as desired, where the final inequality is due to the assumptions that $\|E[U^2|\bS]\|_{P,\infty} \leq C$, $\|E[V_j^2|\bS]\|_{P,\infty} \leq C$ for $j=1,...,p$, and $\|\eta_0-\eta\|_{P,\infty}\leq C$ for $\eta \in \mathcal{T}_n$, and applying the law of iterated expectation.

    Finally we check the condition that $\lambda_n'= \sup_{r \in (0,1), \eta \in \mathcal{T}_n}\|\partial_r^2E[\psi(\bW; \bbeta_0, \eta_0 + r(\eta - \eta_0))]\| \leq \delta'_n/\sqrt{n}$. Define $f(r) := E[\psi(\bW; \bbeta_0, \eta_0 + r(\eta - \eta_0)], r\in(0,1)$. Then for $r \in (0, 1),$
   \begin{align*}
       \delta^2_r f(r) &= E\bigl[ 2(\ell(\bS) - \ell_0(\bS))\times (\bmm(\bS) - \bmm_0(\bS)) - 2(\bmm(\bS)-\bmm_0(\bS))(\bmm(\bS)-\bmm_0(\bS))^T\bbeta_0  \bigr] \\
       &= 2E\bigl[ (\ell(\bS)-\ell_0(\bS))\times(\bmm(\bS)-\bmm_0(\bS))  \bigr] - 2E\bigl[ (\bmm(\bS) - \bmm_0(\bS))(\bmm(\bS) - \bmm_0(\bS))^T\bbeta_0  \bigr]
   \end{align*}

   Then note that:
   \begin{align*}
       \| E\bigl[ (\ell(\bS)-\ell_0(\bS))\times(\bmm(\bS)-\bmm_0(\bS)) \bigr]\| &= \sqrt{ \sum_{j=1}^p E\bigl[ (\ell(\bS) - \ell_0(\bS))\times (m_j(\bS)-m_{0j}(\bS))  \bigr]^2  } \\
       & \leq \sqrt{\sum_{j=1}^pE\bigl[ (\ell(\bS) - \ell_0(\bS))^2  \bigr] \times E\bigl[ (m_j(\bS) - m_{0j}(\bS))^2  \bigr]  } \\
       &= \sqrt{ \sum_{j=1}^p \|\ell - \ell_0\|^2_{P,2} \times \|m_j - m_{0j}\|^2_{P,2} }
   \end{align*}
   For $\eta \in T$, $\|\ell - \ell_0\|^2_{P,2} \times \|m_j - m_{0j}\|^2_{P,2} \leq \delta_n^2n^{-1}$, so:
   \begin{align*}
       \| E\bigl[ (\ell(\bS)-\ell_0(\bS))\times(\bmm(\bS)-\bmm_0(\bS)) \bigr]\| & \leq \sqrt{\sum_{j=1}^p\delta_n^2n^{-1}}\\
       &=\sqrt{p\delta_n^2n^{-1}}\\
       &= \sqrt{p}\delta_nn^{-1/2}
   \end{align*}
   By the same reasoning, $\| E\bigl[ (\bmm(\bS)-\bmm_0(\bS))^2 \bigr]\| \leq \sqrt{p}\delta_nn^{-1/2}$. Then, we have that $\|\partial_r^2E[\psi(\bW; \bbeta_0, \eta_0 + r(\eta - \eta_0))]\| = \| 2E\bigl[ (\ell(\bS)-\ell_0(\bS))\times(\bmm(\bS)-\bmm_0(\bS))  \bigr] - 2E\bigl[ (\bmm(\bS) - \bmm_0(\bS))(\bmm(\bS) - \bmm_0(\bS))^T\bbeta_0  \bigr]\| \leq 2\|E\bigl[ (\ell(\bS)-\ell_0(\bS))\times(\bmm(\bS)-\bmm_0(\bS))  \bigr]\| + 2\|E\bigl[ (\bmm(\bS) - \bmm_0(\bS))(\bmm(\bS) - \bmm_0(\bS))^T\bbeta_0  \bigr]\|$, which does not depend on $r$, so that for $\eta \in T$, $\|\partial_r^2E[\psi(\bW; \bbeta_0, \eta_0 + r(\eta - \eta_0))]\| \leq 4\sqrt{p}\delta_nn^{-1/2} \leq \delta'_n n^{-1/2}$, establishing the desired bound on $\lambda_n'$.
    
    With the conditions of Assumptions 3.1 and 3.2 from \cite{dml} verified, Lemma \ref{theorem:plr} follows from Facts 3.1 and 3.2 from \cite{dml}.
\end{proof}

Lemma \ref{l:bounded_rate} follows directly from Theorem 3.3 from \citep{eberts2013optimal}, using $\rho = \ln(n)$. Note that the estimated function from the least-squares SVM with Gaussian kernel and least-squares loss analyzed in \citep{eberts2013optimal} is identical to the posterior mean of a Gaussian process with Gaussian kernel; see e.g. \citep{kanagawa2018gaussian}.

\begin{lemma}[Convergence rate for bounded regression using GP posterior mean]\label{l:bounded_rate}
    Let $\widehat{f}$ be an estimate of $f(S) = E(Y|S)$ obtained by a Gaussian process posterior mean, with Gaussian kernel, using parameters $\gamma_n, \lambda_n$ selected by the training-validation scheme in \citep{eberts2013optimal} using grids as specified in Algorithm 1. Let $P_S$ be the marginal distribution of $S$ over $\mathbb{R}^d$ with support in the $\|\cdot\|_2$-unit ball. Let the density of $P_S$ be $p_S \in L_q(\mathbb{R}^d)$ for some $q \geq 1$, and let $f \in L_2(\mathbb{R}^d) \cap L_{\infty}(\mathbb{R}^d)$ and $f \in B^\alpha_{2s, \infty}$ for $\alpha \geq 1$ and $s \geq 1$ such that $\frac{1}{q} + \frac{1}{s} = 1$. Let $Y \in [-M, M], M>0$ and let $\widehat{f}$ be clipped at $-M, M$. 
    
    Then with probability no less than $1 - \frac{1}{n}$, $\|\widehat{f} - f\|_{P,2} \leq C \log(n)n^{-\frac{\alpha}{2\alpha + d} + \xi}$, for all $\xi > 0$ and some $C>0$.
\end{lemma}

Lemma \ref{l:unbounded_rate} follows directly from Theorem 3.6 from \citep{eberts2013optimal}, using $\widehat{\rho} = \ln(n)$ and $\overline{\rho} = \ln(n)$ and clipping the absolute value of the fitted function at $\min\{1, M_n\}$ rather than simply $M_n$, which does not change the result since by assumption the true function lies in $[-1, 1]$. Note that per the proof of Theorem 3.6, the constant $C$ in the original statement of Theorem 3.6 does not depend on either $\widehat{\rho}$ or $\overline{\rho}$ allowing these substitutions.

\begin{lemma}[Convergence rate for regression with normal errors using GP posterior mean]\label{l:unbounded_rate}
    Let $\widehat{f}$ be an estimate of $f(S) = E(Y|S)$ obtained by a Gaussian process posterior mean, with Gaussian kernel, using parameters $\gamma_n, \lambda_n$ selected by the training-validation scheme in \citep{eberts2013optimal} using grids as specified in Algorithm 1. Let $P_S$ be the marginal distribution of $S$ over $\mathbb{R}^d$ with support in the $\|\cdot\|_2$-unit ball. Let the density of $P_S$ be $p_S \in L_q(\mathbb{R}^d)$ for some $q \geq 1$, and let $f \in L_2(\mathbb{R}^d) \cap L_{\infty}(\mathbb{R}^d)$ and $f \in B^\alpha_{2s, \infty}$ for $\alpha \geq 1$ and $s \geq 1$ such that $\frac{1}{q} + \frac{1}{s} = 1$. Assume further that $f(S) \in [-1, 1]$.
    
    Let $Y_i = f(S_i) + \epsilon_i$ where $\epsilon_i \sim \text{ ind. } N(0, \sigma_i^2) $, and let there exist some constant $C_0$ such that all $\sigma^2_i < C_0$. Let $\widehat{f}$ be clipped so that $|\widehat{f}| \leq \min\{1, M_n\}$ where $M_n=4\sqrt{C_0}
    \sqrt{\ln(n)}$. 

    Then with probability no less than $1 - \frac{2}{n}$, $\|\widehat{f} - f\|_{P,2} \leq C\log(n) n^{-\frac{\alpha}{2\alpha + 2} + \xi}$, for all $\xi>0$ and some $C>0$.
\end{lemma}

The proof of Theorem 1 follows by verifying the Assumption \ref{assn:plr} using assumptions A1-A6 and Facts 1-3.

\begin{proof}[Proof of Theorem 1]
    First note that estimation of $\bbeta_0$ by $\widehat{\bbeta}_{DSR}$ using Algorithm 1 is equivalent to using DML2 in Definition 3.2 of \cite{dml}, under the Partially Linear Regression DML estimation established by Lemma \ref{theorem:plr}, and estimation of the nuisance parameters uses the same Gaussian Process (GP) estimates used in Facts \ref{l:bounded_rate} and \ref{l:unbounded_rate}. Therefore, the result of Theorem 1 follows from satisfying Assumption \ref{assn:plr}, which in turn is achieved in part by satisfying the assumptions for Facts \ref{l:bounded_rate} and \ref{l:unbounded_rate} which establish the necessary convergence rates for prediction using GP regression.

    In Assumption \ref{assn:plr}, (a) follows by Assumption A1, (b) follows from the assumption of bounded or normal distributions of $U, V_{j}$, so that all moments are finite, and the assumption that $\bbeta_0\in\mathbb{R}^p$, and (c) and (d) follow from Assumption A2.

     To satisfy (e) in Assumption \ref{assn:plr}, first note that by Assumption A4 the components of $\eta_0$ are bounded in some interval, and that the nuisance parameter estimates are clipped accordingly to reside in some interval, satisfying $\|\eta_0 - \widehat{\eta}_0\|_{P,\infty} < C$ for some $C>0$ (recall the notation $\|\eta_0 - \widehat{\eta}_0\|_{P,q} = \max_j\|\eta_{0j} - \widehat{\eta}_{0j}\|_{P,q}$ for $q \in [0, \infty)$). Next apply Lemmas 2 and 3 to satisfy the convergence rate requirements. Let $\alpha_X \geq \frac{d}{2}$ and $\alpha_X >1$ per Assumption A6. Then if each $m_{j0}$ is estimated by $\widehat{m}_{j0}$ using a Gaussian process mean with clipping at appropriate bounds as in Facts \ref{l:bounded_rate} and \ref{l:unbounded_rate}, then with probability no less than $1-\frac{2}{n}$, $\| \widehat{m}_{j0}- m_{j0}\|_{P,2} \leq C_j \log (n) n^{-\frac{\alpha_X}{2\alpha_X + d} + \xi}$ for any $\xi>0$ and some $C_j>0$, and $n^{-\frac{\alpha_X}{2\alpha_X + d}} < n^{-1/4}$. Let $\gamma = \frac{\alpha_X}{2\alpha_X + d} - 1/4 > 0$. Then with probability no less than $1-\frac{2}{n}$, $\| \widehat{m}_{j0}- m_{j0}\|_{P,2} \leq C_j \log (n) n^{-1/4}n^{-\gamma}n^{\xi}$, for all $\xi > 0$. Pick $\xi < \gamma$. Let $\gamma^* = \gamma - \xi > 0$. This holds for $j=1, ..., p$; let $C_m$ be greater than or equal to all $C_j$. Then letting $\delta'_n = C_m(\log(n)n^{-\gamma^*}) \vee n^{-1/4} \rightarrow 0$, we have that $\| \widehat{m}_{j0}- m_{j0}\|_{P,2} \leq \delta'_n n^{-1/4}$ with probability no less than $1-\frac{2}{n}$ for all $m_{0j}$, $j=1, ..., p$, and if $\alpha_Y \geq \frac{d}{2}$ and $\alpha_Y > \frac{d^2}{4\alpha_X}$, the analogous result holds for $g_0$ as well.
     
     Since all of the errors (of estimates of $g_0, m_{01}, ..., m_{0p})$ individually obey the desired rates, each with marginal probability no less than $1 - \frac{2}{n}$, there exists a sequence $\Delta \rightarrow 0$ such that with probability no less than $1 - \Delta_n$, all estimates of $g_0, m_{01}, ..., m_{0p}$ simultaneously obey the desired error bounds. To see why, let $A_{kn}$, $k=1, ..., K$, $n=1,2, ...$ be a finite number $K$ of sequences of events, such that $P(A_{kn}) \rightarrow 1$ as $n \rightarrow \infty$ for each $k$. $P(A_{1n} \cup A_{2n}) = P(A_{1n}) + P(A_{2n}) - P(A_{1n}\cap A_{2n})$, and since $P(A_{1n}\cup A_{2n}) \geq P(A_{1n}), P(A_{2n})$, we have that $P(A_{1n}\cup A_{2n}) \rightarrow 1$ as $n \rightarrow \infty$. Hence, $\lim_{n\rightarrow \infty} P(A_{1n}\cap A_{2n})= \lim_{n\rightarrow \infty} P(A_{1n}) + \lim_{n\rightarrow \infty} P( A_{2n})-\lim_{n\rightarrow \infty} P(A_{1n}\cup A_{2n}) = 1$. Applying induction establishes that $\lim_{n\rightarrow \infty} P(A_{1n}\cap A_{2n}\cap ... \cap A_{Kn})=1$. Therefore, there exists some sequence $L_n \rightarrow 0$ such that $P(A_{1n}\cap A_{2n}\cap ... \cap A_{Kn})\geq 1-L_n$.
     
     Therefore, there exists some sequence $\Delta_n \rightarrow 0$ such that $\|\widehat{g}_0-g_0\|_{P,2} \times \|\widehat{m}_{0j} - m_{0j}\|_{P,2} \leq C \delta_n n^{-\frac{1}{2}}$, $\|\widehat{m}_{0j} - m_{0j}\|^2_{P,2} \leq C \delta_n n^{-\frac{1}{2}}$, and $\|\widehat{\eta}_{0} - \eta_{0}\|_{P,2} \leq \delta_n$ for $j=1, ..., p$, with $\delta_n = \delta_n'^2 \geq n^{-1/2}$, with probability no less than $1 - \Delta_n$.

     Thus part (e) of Assumption \ref{assn:plr} is satisfied.

      With parts (a)-(e) of Assumption \ref{assn:plr} satisfied, apply Lemma \ref{theorem:plr} to obtain $  \widehat{Var}(\widehat{\bbeta}_0)^{-1/2}(\widehat{\bbeta}_0 - \bbeta_0) \xrightarrow[]{D} N(\mathbf{0},\mathbf{I}_p)$.
\end{proof}

 \bibliographystyle{biom} 
\bibliography{refs.bib}